 \documentclass[11pt]{article}

\usepackage{amssymb,amsmath,amsfonts}
\usepackage{graphicx,color,enumitem}
\usepackage{mathrsfs}
\usepackage{amsthm} 
\usepackage[dvipsnames]{xcolor}
\usepackage{bm}
\usepackage[round]{natbib}
\usepackage[]{appendix}
\usepackage{comment}

\usepackage{geometry}

\definecolor{mycolor}{rgb}
{0.137,0.466,0.741}

\RequirePackage[colorlinks,citecolor=mycolor,urlcolor=blue, linkcolor=mycolor]{hyperref}

\usepackage{tikz}
\tikzstyle{vertex}=[circle, draw, inner sep=2pt, fill=white]

\newcommand{\rsto}{]\!\kern-1.8pt ]}
\newcommand{\lsto}{[\!\kern-1.7pt [}

\numberwithin{equation}{section}

\newcommand{\RR}{\mathbb{R}}

\newcommand{\PP}{\mathbb{P}}

\newcommand{\NN}{\mathbb{N}}
\newcommand{\EE}{\mathbb{E}}

\newcommand{\cM}{\mathcal{M}}

\renewcommand{\P}{\mathbb{P}}

\newcommand{\ind}{1\!\kern-1pt \mathrm{I}}

\renewcommand{\rho}{\varrho}

\newcommand\ep{\varepsilon}

\newcommand\E{\mathbb{E}}

\definecolor{mypurple}{rgb}{.50,0,98}

\DeclareMathOperator*{\essinf}{ess\,inf}

\newtheorem{theorem}{Theorem}

\newtheorem{assumption}[theorem]{Assumption}

\newtheorem{corollary}[theorem]{Corollary}

\newtheorem{lemma}[theorem]{Lemma}

\newtheorem{proposition}[theorem]{Proposition}

\theoremstyle{definition}
\newtheorem{definition}[theorem]{Definition}
\newtheorem{remark}[theorem]{Remark}
\newtheorem{example}[theorem]{Example}

\numberwithin{equation}{section}
\numberwithin{theorem}{section}

\begin{document}

\title{Risk measures under model uncertainty: a Bayesian viewpoint}

\author{Christa Cuchiero \thanks{Vienna University, Department of Statistics and Operations Research, Data Science~@~Uni Vienna,
		Kolingasse 14-16, A-1090 Wien, Austria, christa.cuchiero@univie.ac.at}
	\and Guido Gazzani \thanks{Vienna University, Department of Statistics and Operations Research, Kolingasse 14-16, A-1090 Wien, Austria, guido.gazzani@univie.ac.at}
	\and Irene Klein \thanks{Vienna University, Department of Statistics and Operations Research,
		 Kolingasse 14-16, A-1090 Wien, Austria, irene.klein@univie.ac.at.\newline
		The authors gratefully acknowledge financial support from the FWF project I 3852.
}}

\date{}

\maketitle

\begin{abstract} 
We introduce two kinds of risk measures with respect to some reference probability measure, which both allow for a certain order structure and  domination property. Analyzing their relation to each other  leads to the question when a certain minimax inequality is actually an equality. We then provide  conditions under which the corresponding robust risk measures, being defined as the supremum over all risk measures induced by a set of probability measures, can be represented classically in terms of one single probability measure. We focus in 
 particular on the mixture probability measure obtained via mixing 
 over a set of probability measures using some prior, which represents for instance the regulator's beliefs. The classical representation in terms of the mixture probability measure can then be interpreted as a Bayesian approach to robust risk measures.  
\end{abstract}


\tableofcontents

\section{Introduction}\label{sec:intro}
In quantitative risk management, risk measures are used to determine the minimal capital requirement so that a financial position, described by a random variable, becomes acceptable.  
Proceeding from the pioneering works of \cite{ADEH:1999}, \cite{FS:2002} and \cite{FG:2002}, several ideas have been proposed to extend the framework based on a single probability measure to 
a robust setup where the risk evaluation is not only based on one probability measure, but rather on a set of possibly non-dominated ones. A common approach among these is to \emph{robustify} a particular risk measure by considering its worst-case counterpart, meaning to take the supremum over all risk measures induced by the specified set of probability measures. We shall call this \emph{robust (version of a) risk measure.}

 This approach has been pursued  e.g.~for Value at Risk (VaR) by \cite{PYY:2020} and \cite{PWXYY:21}, where the risk measure is robustified via the concept of sublinear expectations, and also by \cite{EPR:13}, where
 the supremum of  VaR over a set of possible joint distributions with prespecified marginals is computed. In the articles by  \cite{PWXYY:21}, \cite{GOO:2003} and \cite{ZKR:2013}  the supremum over all distributions within a certain class of models is considered to compute a `robust' quantile. Similar approaches have been adopted for two particular coherent risk measures, namely Expected Shortfall (ES) and the superhedging price (see e.g., \cite{ZF:2009}, \cite{FBL:2012},  \cite{ALC:2020},  \cite{OW:2021}).
 
Another possible perspective to account for model uncertainty is to modify the functional form of a risk measure by distorting the probability measure or a quantile, see for instance \cite{BGS:2014, WZ:2018} for distorted versions of  VaR or more recently \cite{HPWT:2021} for a distortion of a utility based risk measure.  In \cite{WZ:2018} also different notions of law-invariance for robust versions of VaR and ES are discussed.

These viewpoints have been extensively employed not only in risk evaluation but also in stochastic optimization, see e.g. \cite{WX:2020}, and in insurance, see for instance \cite{EP:2018} and \cite{BP:2019}.

We here introduce a different approach to robustness which is inspired by the following reasoning. Suppose that there is some probability measure $\mathbb{P}$ and an associated  risk measure  which yields for all  financial positions a maximal value, defined e.g.~by the regulator of the financial market as worst case. We shall thus call such a `maximal' measure
 `\emph{worst-case probability measure}' and introduce a certain order structure to incorporate it into a family of risk measures. Indeed, we require 
 (i) that the risk measured under any absolutely continuous measure is smaller or equal than the risk computed under $\mathbb{P}$ and (ii) that for any equivalent measure we get the same risk. Indeed, this worst-case point of view should capture the risk coming from all possible scenarios rather than necessarily taking into account how probable they are, which is of course reminiscent of option pricing under no-arbitrage assumptions. With this approach we aim  to find a classical representation of the robust risk measure in terms of a \emph{single} probability measure, which can then be interpreted as such a `worst-case probability measure'.

To make this program work we need to introduce appropriate (families of) risk measures so that
the requirements (i) and (ii) from above can be fulfilled. In this respect we consider two different ways of defining risk measures with respect to some reference probability measure $P$ (not necessarily equal to $\mathbb{P}$ here).
In both approaches we fix a monetary risk  measure $\rho$ (in the sense of Section \ref{sec:setting} or \cite{FS:2002}) on the space $\mathcal{X}$ of bounded measurable random variables.

 In the \emph{first} approach described in Section~\ref{fixrisk} we then define for a given probability measure $P$  and for all  $X\in L^{\infty}(P)$,
\begin{equation*}
    \rho^{P}(X) :=\inf_{\{\widetilde{X} \in \mathcal{X} : P(\widetilde{X}=X)=1\}}\rho(\widetilde{X}).
\end{equation*}
With this definition it can be easily verified that whenever $P' \ll P $, we have  $\rho^{P'}(X) \leq \rho^{P}(X)$ for all $X \in L^{\infty}(P) $ and thus clearly also $\rho^{P'}(X) = \rho^{P}(X)$ if $P' \sim P$ (see Remark \ref{abs:fixset}). Hence, whenever $\mathbb{P}$ takes the role of $P$, the desired properties (i) and (ii) are satisfied.

In the case where the fixed  risk measure $\rho$ from which we start is convex and continuous from below (see Section \ref{sec:pointwise} for details) we can pursue a \emph{second} approach, outlined in Section \ref{fixpenalty}. 
Indeed, fix some penalty function $\alpha$ on the space of probability measures $\cM_{1}$ and define  $\rho:\mathcal{X}\to\mathbb{R}$ via
\begin{align} \label{eq:rhosecond}
	&\rho(X):=\sup_{Q\in \mathcal{M}_{1}}\{\EE_{Q}[-X]-\alpha(Q)\}, \qquad
	\forall X\in\mathcal{X},
\end{align}
such that it coincides with the given $\rho$, or just use \eqref{eq:rhosecond} as a way to fix a monetary risk measure via a fixed penalty function $\alpha$.
Given a probability measure $P$, the second approach then consists in defining another risk measure $\widehat{\rho}^{P}$ via
\begin{equation*}
    \widehat{\rho}^{P}(X):=\sup_{Q\in\mathcal{P}^{P}}\{\mathbb{E}_{Q}[-X]-\alpha(Q)\}, \quad X\in L^{\infty}(P),
\end{equation*}
where $\mathcal{P}^{P}$ denotes the set of probability measures dominated by $P$. Also this definition yields easily the properties   $\widehat{\rho}^{P'}(X) \leq \widehat{\rho}^{P}(X)$ or any $P'\ll P$ and thus also $\widehat{\rho}^{P'}(X) = \widehat{\rho}^{P}(X)$ if $P' \sim P$  (see Remark \ref{absconsP}). 

Note that the families of risk measures considered in the literature cited above usually do not satisfy these requirements. This holds especially for the approach relying on distortions, which includes as particular cases Expected Shortfall and Value at Risk. Indeed, even for the most basic premium calculation principle (i.e.~risk measure) in life insurance mathematics,  which is just the expected value of the present value of a payment, the order condition is clearly not satisfied.
Similarly, for an arbitrary distortion risk measure, as e.g. Value at Risk, simple examples with $P\sim P'$ and random variables $X_1$, $X_2$ where $P$ gives a smaller value to $X_1$ while $P'$ gives a smaller value to $X_2$
can be constructed to see that the order condition does not hold either.


Coming back to the above defined risk measures $\rho^P$ and $\widehat{\rho}^{P}$, note that for $\rho$ given by \eqref{eq:rhosecond} both $\rho^P$ and $\widehat{\rho}^{P}$ can be defined, whence we aim for a comparison of them 
(see Section \ref{sec:comparision}).
While $\rho^{P}(X)$ is by definition given as the infimum of \eqref{eq:rhosecond}, taken over $P$-almost surely equal claims $\widetilde{X}$, it turns out that $\widehat{\rho}^P$ is equal to the corresponding expression where the infimum and the supremum are interchanged (see Lemma \ref{lem:relation}). Thus for $X\in L^{\infty}(P)$,  $\rho^{P}(X) \geq \widehat{\rho}^P(X)$ in general. 
By means of a classical minimax theorem, we then derive sufficient conditions under which they coincide, see  Proposition \ref{th:minimax}, but also provide an example where the strict inequality holds (see Example \ref{ex:nonequal}).

This comparison is continued in Section \ref{sec:comparisionlawinv}, where we specialize the previous setting further, assuming that the initial risk measure  $\rho$ is law-invariant with respect to some fixed probability measure, denoted here by $\mathbb{P}$, 
whose significance will become clear below.
Starting from this law-invariant risk measure we proceed as above and define $\rho^{P}$ and $\widehat{\rho}^{P}$, for some  probability measure $P$ and show that under certain assumptions, depending on the relation between $\mathbb{P}$ and $P$,  $\rho^{P}$ and $\widehat{\rho}^{P}$ 
coincide (see Proposition \ref{prop:PvsR}, Proposition
\ref{lemma_dual_regulator} and Corollary \ref{th:minimax1}).
The economically most interesting aspect of this setup is that  the
starting probability measure $\mathbb{P}$ with  respect to which law-invariance is considered turns out to qualify as `worst case probability measure'. Indeed, the risk computed under $\mathbb{P}$ dominates the risk computed under any other probability measure $P$, no matter in which relation $\mathbb{P}$ and $P$ are (see Proposition \ref{prop:PvsR} for details).

Having introduced the above setup that allows to consider families of risk measures, i.e., either $(\rho^P)_{P \in \mathcal{M}}$ or $(\widehat{\rho}^P)_{P \in \mathcal{M}}$, induced by some family $\mathcal{M}$ of probability measures, we consider
in Section \ref{sec:robust} and Section \ref{sec:mixture} the corresponding robust versions, i.e.
\begin{align}
\sup_{P \in \mathcal{M}} \rho^P \quad  \text{or} \quad \sup_{P \in \mathcal{M}}\widehat{ \rho}^P. \label{eq:worstcaserisk}
\end{align}
 Observe that these can be seen as particular \emph{generalized risk measures}, introduced recently in \cite{FLW:2021}. 
 As mentioned above our main goal is to link the robust risk measures given by \eqref{eq:worstcaserisk} with a \emph{single} `worst case probability measure' $\mathbb{P}$ and to express  \eqref{eq:worstcaserisk} via $\rho^{\mathbb{P}}$ or $\widehat{\rho}^{\mathbb{P}}$ respectively. Indeed, as we prove in Theorem \ref{th:nonhat} and Theorem \ref{prop:criteria} such a classical representation is possible if $\mathcal{M}$ is closed under countable convex combinations and if there is a measure $\mathbb{P}$ that satisfies the condition of \emph{generalized equivalence} (see Definition \ref{abscon}). The latter means in particular that every measure in $\mathcal{M}$ can be dominated by $\mathbb{P}$. Even though any countable set of probability measures can be dominated, this can be  considered a strong condition, as it excludes  uncountable families of mutually singular measures. Note, however, that this rather strong domination property \emph{would still not} be enough to get a similar result for other robust versions of risk measures, induced e.g.~via distortions, since they do not satisfy the required order structure.
Indeed, this generalized equivalence property only works in combination with the ordering requirements  from above, which led to the specification of the considered family of risk measures being crucial for this result.

In the setup of Section \ref{sec:comparisionlawinv}, where the starting measure $\mathbb{P}$ with respect to which law-invariance is considered qualifies as `worst case probability measure' we get a classical representation without any assumption except that $\mathbb{P}$ has to lie in $\mathcal{M}$. In particular, the \emph{fully non-dominated case} can be considered in this setup. The reason why this works here is because the ordering is achieved on the level of the risk measures rather than the probability measures.

 In Section \ref{sec:mixture} we then analyze a particular candidate for a `worst case probability measure', giving rise to a \emph{Bayesian point of view} to robust risk measures. Indeed, let $\mathcal{M}=\{P^{\theta}: \theta \in \Theta\}$ be a family of probability measures defined via a parameter space $\Theta$, on which some (prior) distribution $\nu$ is specified to model for instance the regulator's beliefs. Consider then as candidate for the  `worst case probability measure', the so-called \emph{mixture probability measure,} given by 
 \[
 \mathbb{P}(\cdot)= \int_{\Theta} P^{\theta}(\cdot) \nu(\mathrm{d}\theta).
 \]
 We then show that under the (rather strong) assumption of  continuity in total variation of   $\theta \mapsto P^{\theta}$, this mixture probability measure actually qualifies as `worst case probability measure' in the sense that it is a generalized equivalent measure  for a subset of $\mathcal{M}$,  induced from a subset of  $\Theta$ that has $\nu$-measure $1$. Even though this continuity assumption implies the existence of a dominating measure (see Lemma \ref{firstcountableprior}), this measure is not  necessarily equivalent to the mixture probability measure $\mathbb{P}$ (see Remark \ref{rem:dommeasures}) and can thus potentially yield  higher values for the associated risk. 
 Taking instead the mixture measure as the `worst case probability measure' means to take the (regulator's) beliefs expressed via $\nu$ seriously. As above we  obtain a certain classical representation for the robust risk measures, now via the mixture probability measure (see Theorem \ref{th:mainmixture}).
 In the setup of Section \ref{sec:comparisionlawinv}, we can again consider a  non-dominated situation to get, under mild conditions, a classical representation with respect to the mixture probability measure $\mathbb{P}$, where in this case $\mathbb{P}$ does not even have to lie in $\mathcal{M}$ (see Corollary \ref{cor:final}).
 Summarizing our main contributions are twofold: 
\begin{enumerate}
    \item[(1)] 
 we introduce two families of risk measures allowing for a certain domination property and analyze their relation to each other leading to the question when a certain mini-max inequality is actually an equality. In the case of law-invariance our approaches naturally lead to a `worst case probability measure' under which the risk is always maximal.
\item[(2)] 
 we show under which conditions a classical representation of the robust risk measures in terms of single probability holds, where a particular focus is on the mixture probability measure giving rise to  a Bayesian point of view on robust risk measures.  
\end{enumerate}

The remainder of the paper is organized as follows.  In Section \ref{sec:setting} we recall frequently used definitions and properties of monetary and convex risk measures following  \cite{FS:2002, FS:2011}. In Section \ref{sec:rm_wrt_referencemeasures} we 
introduce the two different ways of how to define families of risk measures and  discuss their relation to each other.
In Section \ref{sec:examples}, 
all these comparison results are illustrated via some explanatory examples of classical risk measures, e.g.~the superhedging price or Expected Shortfall. Section \ref{sec:robust} establishes conditions for classical representation of robust risk measures in terms of one single `worst case probability measure'. In Section \ref{sec:mixture} this is then applied to the mixture probability measure where we analyze when it can serve as such a  `worst case probability measure'.

\section{Preliminaries on risk measures}\label{sec:setting}

This section is dedicated to recall some basic concepts and definitions of risk measures. After introducing some notation, we  start here with risk measures being defined in a pointwise way (see Section \ref{sec:pointwise}).

\subsection{Notation}

Let $(\Omega, \mathcal{F})$ be a measurable space and let
$\cM_{1}=\cM_{1}(\Omega,\mathcal{F})$ be the set of probability measures
defined on it. We equip it with the topology of weak convergence (in the probabilistic sense), which corresponds  to the functional analytic notion of weak-$*$-convergence when $\Omega$ is compact.
Our aim
is to quantify the risk of a random variable $X$ on $(\Omega,\mathcal{F})$ that describes a certain financial risky position, e.g.~the loss of an
insurance contract or a portfolio.
As a first step we will look at risky claims taking values in the following set:
\begin{equation*}
	\mathcal{X}=\{X:(\Omega,\mathcal{F})\to(\mathbb{R},\mathcal{B}(\RR)) \text{, measurable and bounded}\},
\end{equation*}
where $\mathcal{B}(\mathbb{R})$ denotes the Borel-$\sigma$-algebra on $\mathbb{R}$.
By ``bounded" we mean that for each $X\in\mathcal{X}$ there exists $K>0$ such that $|X(\omega)|\leq K$,  for \emph{all} $\omega\in\Omega$, i.e.
$\mathcal{X}=L^{\infty}(\Omega,\mathcal{F})$.
Later we will relax the condition to an almost sure condition. If a probability measure $P\in\cM_{1}$ is given, we will write $L^{\infty}(P)$ for the space $L^{\infty}(\Omega,\mathcal{F},P)$.

\subsection{Pointwise definition of risk measures}\label{sec:pointwise}
Let us recall in the following the definition of monetary risk measures defined
on the space $\mathcal{X}$.

\begin{definition}[Monetary risk measure]\label{monetary_rm}
	A map $\rho:\mathcal{X}\to\RR$ is said to be a monetary risk measure if:
	\begin{enumerate}
		\item $X\ge0$ then $\rho(X)\le0$, for all $X\in\mathcal{X}$;
		\item $X\ge Y$ then $\rho(X)\le\rho(Y)$, for all $X, Y \in
		\mathcal{X}$;
		\item $\rho(a+X)=\rho(X)-a$, for every $a\in\mathbb{R}$ and  for all
		$X\in\mathcal{X}$.
	\end{enumerate}
\end{definition}
 
Notice that we mean that each property holds for all $\omega\in\Omega$.\\
Given a monetary risk measure $\rho$ the corresponding acceptance set is
defined as follows:
\begin{equation}\label{acceptmeasfree}
	C=\{X\in\mathcal{X}: \rho(X)\leq 0\}.
\end{equation}
 
Note that we get $\rho$ back from the acceptance set via the equation
\begin{equation*}
	\rho(X)=\inf\{m\in\mathbb{R}: m+X\in C\}.
\end{equation*}
Let us also state the definition of a convex risk measure, where the convexity assumption is economically meaningful since
diversification should not increase the risk of a financial position. We also recall the main results
concerning their dual representation.
\begin{definition}[Convex risk measure]\label{convex_rm}
	A map $\rho:\mathcal{X}\to\RR$ is said to be a convex risk measure if it is a monetary risk measure and if for every $\lambda\in(0,1)$ and for all $X,Y\in\mathcal{X}$
 $$\rho(\lambda X + (1-\lambda) Y)\le
		\lambda\rho(X)+(1-\lambda)\rho(Y).$$
\end{definition}
The acceptance set of a convex risk measure is defined exactly as for
monetary risk measures in \eqref{acceptmeasfree}.\\
Assuming moreover, that a 
convex risk measure $\rho$ on
$\mathcal{X}$ 
is continuous from below in the sense that for any sequence $(X_{n})_{n\in\mathbb{N}}\subset \mathcal{X}$,
\begin{align}\label{eq:contbelow}
X_n \uparrow X \text{pointwise} \quad \Rightarrow \quad \rho(X_n) \downarrow \rho (X),
\end{align}
as $n\to+\infty$, then we obtain the following
dual representation 
\begin{equation}\label{dual_representation}
	\rho(X)=\sup_{Q\in \mathcal{M}_{1}}\{\EE_{Q}[-X]-\alpha(Q)\},\qquad \forall
	X\in\mathcal{X},
\end{equation}
where $\alpha:\mathcal{M}_{1}\to\RR\cup\{+\infty\}$ is a penalty function, see Proposition 4.21 in \cite{FS:2011}.
The penalty function itself it is not unique in general. Note however that due to the continuity from below any penalty function is concentrated on $\mathcal{M}_1$ instead of finitely additive set functions denoted by $\mathcal{M}_{1,f}$, meaning that  $\alpha(Q) = \infty$ for all $Q \in \mathcal{M}_{1,f} \setminus \mathcal{M}_1$ so that it is sufficient to take the supremum in \eqref{dual_representation} only over $\mathcal{M}_1$. The penalty function can be chosen
to attain the smallest values by setting
\begin{equation}\label{minimal_penalty}
\begin{split}
	\alpha_{\min}(Q):&=\sup_{X\in C}\EE_{Q}[-X]=\sup_{X \in \mathcal{X}} \{\E_Q[-X]- \rho(X)\}, \qquad  \forall Q\in\cM_{1},
	\end{split}
\end{equation}
i.e.~$\alpha_{\min}$ corresponds to the Fenchel–Legendre transform or conjugate function of the convex function $\rho$ on $\mathcal{X}$,
see Theorem 4.15 as well as Remark 4.16 and Remark 4.17 in \cite{FS:2011}. 

Finally let us recall the notion of coherent risk measures.

\begin{definition}[Coherent risk measure]
  A convex risk measure $\rho : \mathcal{X} \to \mathbb{R}$ is called coherent risk measure if it
satisfies positive homogeneity, i.e. for $\lambda \geq 0$ and all $X\in\mathcal{X}$, $\rho(\lambda X) = \lambda \rho(X)$.
\end{definition}

Note that by Corollary 4.18 in \cite{FS:2011} 
the minimal penalty function of a coherent measure that is continuous from below always takes the form
\begin{align}\label{eq:alphacoh}
\alpha(Q)=\begin{cases} 0  & \text{if} \ Q \in \mathcal{Q},\\
+ \infty  & \text{otherwise},
\end{cases}
\end{align}
for some convex set $\mathcal{Q} \subseteq \mathcal{M}_1$. For the minimal penalty function the set $\mathcal{Q}$ is maximal. Note that  it could potentially be also represented by another $\alpha$ of  form \eqref{eq:alphacoh} on a smaller set $\mathcal{Q}$ (see Proposition 4.14 in \cite{FS:2011}).

\section{Risk measures with respect to reference probability measures}\label{sec:rm_wrt_referencemeasures}

From now on, we additionally suppose that we are given a probability  measure $P$ on $(\Omega,\mathcal{F})$ with respect to which we shall define risk measures such that the domination properties mentioned in the introduction hold true. We shall pursue two different approaches. In both of them we start by fixing a monetary risk  measure $\rho$ on the space $\mathcal{X}$ of bounded measurable random variables. We then define two possible ways how to deduce from $\rho$ a risk measure depending on $P$. The first approach works for general monetary risk measures, whereas the second assumes convexity.

In Section \ref{sec:comparision} and Section \ref{sec:comparisionlawinv}, we shall then compare these two approaches in the case of convex as well a convex and law-invariant risk measures.

\subsection{First approach applicable to general monetary risk measures}\label{fixrisk}

The first approach consists in defining a risk measure $\rho^P$ as follows:

\begin{definition}[Risk measure $\rho^{P}$]\label{rhoP}
	Let $P\in\cM_{1}$ and $\rho$ a monetary risk measure. Then
	 the corresponding risk measure with respect to $P$ is  denoted by $\rho^P: L^{\infty}(P) \to \mathbb{R}$ and defined by
	\begin{align*}
		\rho^{P}(X) &=\inf_{\{\widetilde{X} \in \mathcal{X} : P(\widetilde{X}=X)=1\}}\rho(\widetilde{X}),
	\end{align*}
	for all $X\in L^{\infty}(P)$.
	Moreover, the corresponding acceptance set is given by
\begin{align*}
C^P &=\{X \in L^{\infty}(P) : \inf_{\{\widetilde{X} \in \mathcal{X} : P(\widetilde{X}=X)=1\}}\rho(\widetilde{X}) \leq 0 \}\\
&=\{X \in L^{\infty}(P) : \rho^{P}(X)\leq0\}.
\end{align*}	
\end{definition}

Observe that the definition of $C^{P}$ implies that $\rho^{P}$
satisfies the properties in
Definition~\ref{monetary_rm} \emph{not}
for all
$\omega\in\Omega$, but just for the scenarios $\omega\in\Omega$ that occur
$P$-a.s.
\begin{remark}\label{abs:fixset}
	Let $P,P'\in\cM_{1}$ be such that $P'\ll P$. Then,  clearly $$\{\widetilde{X} \in \mathcal{X} : P(\widetilde{X}=X)=1\}\subseteq\{\widetilde{X} \in \mathcal{X} : P'(\widetilde{X}=X)=1\},$$ thus $C^{P}\subseteq C^{P'}$ and
	\begin{equation*}
		\rho^{P'}(X)\le \rho^{P}(X), \qquad \forall X\in L^{\infty}(P).
	\end{equation*}
	In particular, if $P \sim P'$, we have $\rho^{P'}(X)= \rho^{P}(X)$. When the `worst case risk measure' takes the role of $P$, this leads to the desired properties (i) and (ii) mentioned in  the introduction.
\end{remark}
 
The following lemma states that the property of being a monetary or convex risk measure also translates to $\rho^P$.

\begin{lemma}\label{ineherit_prop}
If $\rho$ is a monetary (convex respectively) risk measure, then $\rho^P$ is a monetary (convex respectively) risk measure.
\end{lemma}

\begin{proof}
 We will prove the properties of a monetary risk measure $P$-a.s., namely $\text{(i)-(iii)}$ and finally address the convexity here denoted by $\text{(iv)}$.
 \begin{itemize}
     \item[(i)] Let $X\geq 0$ $P$-a.s. We can choose a representant $\bar{X}=X\ind_{\{X\geq0\}}$, then $P(X=\bar{X})=1$. $\bar{X}(\omega)\geq 0$ for all $\omega$, hence $\rho(\bar{X})\leq0$. By definition it follows that $\rho^P(X)\leq\rho(\bar{X})\leq0$. 
     \item[(ii)] Let $X,Y\in L^{\infty}(P)$ such that $X\geq Y$ $P$-a.s. Since $\rho^P(Y)=\inf_{\{\widetilde{Y} \in \mathcal{X} : P(\widetilde{Y}=Y)=1\}}\rho(\widetilde{Y})$, for each $n\geq 1$, we can choose $\widetilde{Y}^n \in \mathcal{X}$, where $P(\widetilde{Y}^n=Y)=1$ and  such that $\rho(\widetilde{Y}^n)\leq \rho^P(Y)+\frac1{n}$. Consider now any $\widetilde{X}$ with $P(\widetilde{X}=X)=1$ and modify this for each $n$ as follows:
$$\widetilde{X}^n=X\ind_{\{\widetilde{Y}^n=Y\} \cap \{ X \geq Y\}}+\max\{\widetilde{X}, \widetilde{Y}^n\}\ind_{\{\widetilde{Y}^n\neq Y\} \cup \{ X < Y\} }.$$
Clearly, $\widetilde{X}^n\in\mathcal{X}$ and $P(\widetilde{X}^n=X)= P((\widetilde{Y}^n=Y) \cap (X \geq Y) )=1$, as $P(X\geq Y)=1$. Then we get, for all $\omega$,
$$\widetilde{X}^n(\omega)\geq Y\ind_{\{\widetilde{Y}^n=Y\} \cap \{ X \geq Y\}}(\omega)+ \widetilde{Y}^n \ind_{\{\widetilde{Y}^n\neq Y\} \cup \{ X < Y\}}(\omega)= \widetilde{Y}^n(\omega).$$
Therefore as $\rho$ satisfies (ii) we get $\rho(\widetilde{X}^n)\leq\rho(\widetilde{Y}^n)\leq \rho^P(Y)+\frac1{n}$. This implies
$$\rho^P(X)\leq \inf_{n\geq1} \rho(\widetilde{X}^n)\leq \rho^P(Y),$$
and (ii) follows for $\rho^P$.
\item[(iii)] Observe that the property of cash-invariance or translation-invariance holds trivially as for all $X\in L^{\infty}(P)$ and $a\in\mathbb{R}$ we have
\begin{equation*}
    \rho^{P}(X+a)=\inf_{\{\widetilde{X} \in \mathcal{X} : P(\widetilde{X}=X+a)=1\}}\rho(\widetilde{X})=\inf_{\{\widetilde{X} \in \mathcal{X} : P(\widetilde{X}=X)=1\}}\rho(\widetilde{X}+a)=\rho^{P}(X)-a.
\end{equation*}
\item[(iv)]Let us now assume that $\rho$ is convex. Suppose $X$ and $Y$ are in $L^{\infty}(P)$ and $\lambda\in (0,1)$. Let $\widetilde{X}$, $\widetilde{Y}$ be in $\mathcal{X}$ such that $P(X=\widetilde{X})=P(Y=\widetilde{Y})=1$. Define $\bar{Z}=\lambda \widetilde{X}+(1-\lambda)\widetilde{Y}$. Then, clearly, $\bar{Z}\in\{\widetilde{Z}\in\mathcal{X}: P(\widetilde{Z}=\lambda X+(1-\lambda)Y)=1\}$. It follows from the convexity of $\rho$ that
\begin{align*}
\rho^P(\lambda X+(1-\lambda)Y) &\leq \inf_{\{\bar{Z}\in\mathcal{X}:\exists \widetilde{X}, \widetilde{Y}: \bar{Z}=\lambda \widetilde{X}+(1-\lambda)\widetilde{Y}\}}\rho(\bar{Z})\\
&\leq \inf_{\{(\widetilde{X}, \widetilde{Y})\in\mathcal{X}\times\mathcal{X}: P(X=\widetilde{X})=P(Y=\widetilde{Y})=1\}} \{\lambda \rho(\widetilde{X})+(1-\lambda)\rho(\widetilde{Y})\}\\
&=\lambda \rho^P(X)+(1-\lambda)\rho^P(Y),
\end{align*}
showing that $\rho^P$ is convex.
 \end{itemize}
 \end{proof}

\subsection{Second approach applicable to convex risk measures: fixing a penalty function}\label{fixpenalty}
Let us now turn to the second approach  
 how to  introduce risk measures depending on some reference measure, 
 which works in the case of convex risk measures that are continuous from below in the sense of \eqref{eq:contbelow}.
 Indeed, fix some penalty function $ \alpha:\mathcal{M}_{1}\to\RR\cup\{+\infty\}$ and define a convex risk
measure $\rho:\mathcal{X}\to\RR$ via
\begin{align}\label{eq:convexriskm}
	&\rho(X):=\sup_{Q\in \mathcal{M}_{1}}\{\EE_{Q}[-X]-\alpha(Q)\}, \qquad
	\forall X\in\mathcal{X}.
\end{align}

Here, instead of looking at the risk measure $\rho^P$,  we  modify
$\alpha$, which then gives rise to a new risk measure $\widehat{\rho}^P$,
 depending on the set of absolutely continuous measures with respect to $P$, denoted by
\[
\mathcal{P}^{P}=\{Q\in\cM_{1} : Q\ll P\}.
\]
To be precise, let $\alpha$  be the fixed penalty function as of \eqref{eq:convexriskm}
and define $\widehat{\rho}^P: L^{\infty}(P) \to \mathbb{R}$ as follows:
\begin{align}\label{eq:rhohat}
\widehat{\rho}^{P}(X) =\sup_{Q\in
	\mathcal{P}^{P}}\{\EE_{Q}[-X]-\alpha(Q)\}.
\end{align}
This means that we actually have to use the following penalty function $\alpha^P(Q)$ in order to express the supremum again over all probability measures $\mathcal{M}_1$. Indeed,
define
\begin{equation}\label{alphaP}
	\alpha^{P}(Q):=
	\begin{cases}
		\alpha(Q) \quad	&\text{for} \ Q\in\mathcal{P}^{P}, \\
		+\infty \quad &\text{for} \ Q\in\mathcal{M}_{1}\setminus
		\mathcal{P}^{P},
	\end{cases}
\end{equation}
then we exactly get
\begin{align}\label{rhohat}
	\widehat{\rho}^{P}(X) =\sup_{Q\in
		\mathcal{M}_{1}}\{\EE_{Q}[-X]-\alpha^{P}(Q)\}, \quad \forall X \in L^{\infty}(P)
\end{align}

As before the corresponding
acceptance set of $\widehat{\rho}^P$ is defined as
\begin{equation}\label{CHat}
	\widehat{C}^{P}=\{X\in L^{\infty}(P):
	\widehat{\rho}^{P}(X)\le0\}.
\end{equation}

\begin{remark}\label{absconsP}
	Let $P,P'\in\cM_{1}$ such that $P'\ll P$. Then, obviously, if $Q\ll P'$ then $Q\ll P$ as well. Hence $\mathcal{P}^{P'}\subseteq \mathcal{P}^{P}$. This implies that for all $X\in L^{\infty}(P)$
	$$\widehat{\rho}^{P'}(X) =\sup_{Q\in
		\mathcal{P}^{P'}}\{\EE_{Q}[-X]-\alpha(Q)\}\leq \sup_{Q\in
		\mathcal{P}^{P}}\{\EE_{Q}[-X]-\alpha(Q)\}=\widehat{\rho}^{P}(X).$$
In particular, if $P \sim P'$, then $\widehat{\rho}^{P}= \widehat{\rho}^{P'}$. When the `worst case risk measure' takes the role of $P$, this leads again to the desired properties (i) and (ii) mentioned in  the introduction.
\end{remark}

\subsection{Comparing the two approaches for convex risk measures} \label{sec:comparision}

Our goal is now to compare these two different approaches in the case of convex risk measures and establish conditions under which they coincide. To this end we let $\rho$  be the convex risk measure given by \eqref{eq:convexriskm} for some fixed $\alpha$.

According to the first approach define now for $X\in L^{\infty}(P)$  
\[
\rho^P(X)=\inf_{\{\widetilde{X} : P(\widetilde{X}=X)=1\}}\rho(\widetilde{X})= \inf_{\{\widetilde{X} : P(\widetilde{X}=X)=1\}} \sup_{Q \in \mathcal{M}_1}\{\mathbb{E}_Q[\widetilde{X}]-\alpha(Q)\}.
\]
The obvious question at that point  is  when  does the following minimax inequality
\begin{equation}
\begin{split} \label{eq:duality}
\widetilde{\rho}^P(X):&=\sup_{Q \in \mathcal{M}_1}\inf_{\{\widetilde{X} \in \mathcal{X} : P(\widetilde{X}=X)=1\}} \{\mathbb{E}_Q[-\widetilde{X}]-\alpha(Q)\}\\
&\leq \inf_{\{\widetilde{X} \in \mathcal{X} : P(\widetilde{X}=X)=1\}} \sup_{Q \in \mathcal{M}_1}\{\mathbb{E}_Q[-\widetilde{X}]-\alpha(Q)\} =\rho^P(X).
\end{split}
\end{equation}
reduce to an equality. Before clarifying this, we shall first relate these risk measures with $\widehat{\rho}^{P}$ given by \eqref{eq:rhohat} or equivalently by $\eqref{rhohat}$ . Indeed, the following lemma  shows that $\widehat{\rho}^P$ coincides with $\widetilde{\rho}^P$ but that we have a potential strict inequality with $\rho^P$.

\begin{lemma} \label{lem:relation} For $X\in L^{\infty}(P)$, we have  $\rho^P(X) \geq  \widetilde{\rho}^P(X) = \widehat{\rho}^P(X)$
and thus $C^{P}\subseteq\widehat{C}^{P}$.
\end{lemma}

\begin{proof}
Let $X \in L^{\infty}(P)$. Then for all $\widetilde{X}\in \mathcal{X}$ with $P(\widetilde{X}=X)=1$ and all $Q\in \mathcal{P}^P$
\[
\EE_{Q}[-X]-\alpha(Q)=\EE_{Q}[-\widetilde{X}]-\alpha(Q).
\]
Hence, also
\[
\EE_{Q}[-X]-\alpha(Q)=\inf_{\{\widetilde{X} : P(\widetilde{X}=X)=1\}}  \EE_{Q}[-\widetilde{X}]-\alpha(Q).
\]
Therefore we get the following inequalities
\begin{align*}
\widehat{\rho}^{P}(X)&=\sup_{Q\in\mathcal{P}^{P}}\{\EE_{Q}[-X]-\alpha(Q)\}=\sup_{Q\in\mathcal{P}^{P}}\inf_{\{\widetilde{X} : P(\widetilde{X}=X)=1\}}\{\EE_{Q}[-\widetilde{X}]-\alpha(Q)\}\\
&\leq \underbrace{\sup_{Q\in\cM_{1}}\inf_{\{\widetilde{X} : P(\widetilde{X}=X)=1\}}\{\EE_{Q}[-\widetilde{X}]-\alpha(Q)\}}_{\widetilde{\rho}^P(X)}\leq \underbrace{\inf_{\{\widetilde{X} : P(\widetilde{X}=X)=1\}}\sup_{Q\in\cM_{1}} \{\EE_{Q}[-\widetilde{X}]-\alpha(Q)\}}_{\rho^P(X)}.
\end{align*}
Hence if $X \in C^P$, meaning that $\rho^P(X) \leq 0$, then also $\widehat{\rho}^P(X) \leq 0$ so that $ C^P \subseteq \widehat{C}^P$.

Let us now prove the missing equality
$
\widetilde{\rho}^P = \widehat{\rho}^P.
$
Let $Q \in \cM_1$ and consider the Lebesgue decomposition into $Q=\widetilde{Q}_1 +\widetilde{Q}_2$ where $\widetilde{Q}_1$ and $\widetilde{Q}_2$ are nonnegative measures such that $\widetilde{Q}_1 \ll P$ and $\widetilde{Q}_2 \perp P$. Moreover denote by $Q_i$ the probability measure $Q_i:= \frac{\widetilde{Q}_i}{\widetilde{Q}_i(\Omega)}$ for $i=1,2$. Then for every fixed $Q$ with $\widetilde{Q}_2(\Omega) >0$, there is $B_Q\in\mathcal{F}$ with $\widetilde{Q}_2(B_Q)=\widetilde{Q}_2(\Omega)$ and $P(B_Q^c)=1$. Define, for every $n\geq1$,
$\widetilde{X}^n_Q=n\ind_{B_Q}+X\ind_{B_Q^c}\in\mathcal{X}$. Clearly $P(\widetilde{X}^n_Q=X)=P(B_Q^c)=1$. We have that
$$\E_Q[-\widetilde{X}^n_Q]=-nQ(B_Q)+\E_{Q}[-X\ind_{B_Q^c}]=-n\widetilde{Q}_2(\Omega)+\widetilde{Q}_1(\Omega)\E_{Q_1}[-X].$$
Moreover we know that $\inf_{R\in\mathcal{M}_1}\alpha(R)>-\infty$, say $-N:=\inf_{R\in\mathcal{M}_1}\alpha(R)$. Note that, as $X\in L^{\infty}(P)\subseteq L^{\infty}(Q_1)$ we have that
 $|\E_{Q_1}[-X]|\leq \|X\|_{L^{\infty}(P)}$. It follows that
\begin{align*}
\inf_{\{\widetilde{X} : P(\widetilde{X}=X)=1\}}\{\E_Q[-\widetilde{X}]-\alpha(Q)\} &\leq \inf_{n\geq 1}    \{\E_Q[-\widetilde{X}^n_Q]-\alpha(Q)\}\\
&\leq \inf_{n\geq1}-n\widetilde{Q}_2(\Omega)+\widetilde{Q}_1(\Omega)\E_{Q_1}[-X]+N=-\infty.
\end{align*}

Taking the supremum over $Q \in \cM_{1}$ therefore just means to restrict the set to $\mathcal{P}^P$ so that $\widetilde{Q}_2\equiv 0$. Hence, 
\begin{align*}
\widetilde{\rho}^P(X)&=\sup_{Q\in\cM_{1}}\inf_{\{\widetilde{X} : P(\widetilde{X}=X)=1\}}\{\EE_{Q}[-\widetilde{X}]-\alpha(Q)\}\\&=\sup_{Q\in \mathcal{P}^P}\inf_{\{\widetilde{X} : P(\widetilde{X}=X)=1\}}\{\EE_{Q}[-\widetilde{X}]-\alpha(Q)\}=\widehat{\rho}^P(X),
\end{align*}
yielding the assertion.

\end{proof}

By the above lemma, equality of $\widehat{\rho}^P=\rho^P$ is thus equivalent to a minimax equality. In the following we thus apply  a (standard) minimax theorem to conclude equality of $\widehat{\rho}^P=\widetilde{\rho}^P=\rho^P$.

\begin{proposition}\label{th:minimax}
Let $\Omega$ be compact and $Q \mapsto \alpha(Q)$  lower semicontinuous  with respect to the weak-$*$-topology on $\mathcal{M}_1$ and convex.
Then
the minimax equality
\begin{align*}
\widetilde{\rho}^P(X)&=\sup_{Q \in  \mathcal{M}_1}\inf_{\{\widetilde{X} \in \mathcal{X} : P(\widetilde{X}=X)=1\}} \{\mathbb{E}_Q[-\widetilde{X}]-\alpha(Q)\}\\
&= \inf_{\{\widetilde{X}   \in \mathcal{X}: P(\widetilde{X}=X)=1\}} \sup_{Q \in  \mathcal{M}_1}\{\mathbb{E}_Q[-\widetilde{X}]-\alpha(Q)\} =\rho^P(X),
\end{align*}
holds true. Hence, under these conditions $\widehat{\rho}^P= \rho^P$.
\end{proposition}

\begin{proof}
Note that
\[
Q \mapsto \E_Q[\widetilde{X}] - \alpha(Q)
\]
is concave and upper semicontinous and
\[
\widetilde{X} \mapsto \E_Q[\widetilde{X}] - \alpha(Q)
\]
is linear and hence convex. Moreover, $\{\widetilde{X} \in \mathcal{X} : P(\widetilde{X}=X)=1\}$ and $\mathcal{M}_1$ are convex. By compactness of $\Omega$, $\mathcal{M}_1$ is also (weak-$*$-) compact. Therefore Sion's minimax theorem (see e.g.~\cite{S:58, RS:13}) yields the first assertion and the second one follows from Lemma \ref{lem:relation}.
\end{proof}

\begin{remark}\label{rem:minalpha}
Note that if $\Omega$ is compact and $\alpha$ happens to be $\alpha_{\min}$ as introduced in Section \ref{sec:pointwise}, all conditions of Proposition \ref{th:minimax}  are automatically satisfied.
Indeed, we can argue along the lines of Remark 4.17 in \cite{FS:2011} and  apply the general duality
theorem for the Fenchel–Legendre transform to the convex function $\rho$, however here on the Banach space of continuous functions on the compact set  $\Omega$, denoted by $C(\Omega)$. This allows  to conclude that 
\[
\alpha_{\min}(Q)=\sup_{X \in C(\Omega)} \{\E_Q[-X]-\rho(X)\}.
\]
Then due to Theorem 2.3.1 in \cite{Z:02}, $\alpha_{\min}$ is convex and weak-$*$-lower semicontinuous. 
\end{remark}

In the  following example  lower semicontinuity of  $\alpha$ fails and this shows
that in general  $C^{P}$ and $\widehat{C}^{P}$
do not coincide, namely $C^{P}$ is strictly contained in $\widehat{C}^{P}$ but not equal. Hence there is a gap between $\widehat{\rho}^P$ and $\rho^P$, meaning that $\widehat{\rho}^P= \widetilde{\rho}^P$ can in general take strictly smaller values than $\rho^P$.

\begin{example}\label{ex:nonequal}
	Let $(\Omega,\mathcal{F})=([0,1],\mathcal{B}([0,1]))$ and the reference
	measure $P$ be the Lebesgue measure restricted to $\mathcal{B}([0,1])$.
	We fix a penalty function $\alpha:\cM_{1}\to\mathbb{R}$ such that
	\begin{align*}
		\alpha(Q):=\begin{cases}
			0, &\text{if} \ Q\in\mathcal{P}^{P},\\
			-1, &\text{if} \ Q\in\cM_{1}\setminus\mathcal{P}^{P}.
		\end{cases}
	\end{align*}
	
	Note that $\alpha$ fails to be lower semicontinuous. Indeed, take a sequence of measures $Q_n=\sum_{i=1}^n \alpha_i^n \delta_{x_i^n}$, where  $\alpha_i^n \geq 0$, $\sum_{i=1}^n \alpha_i^n=1$ and $\delta_{x}$ is the Dirac measure at $x$, converging to some measure $Q$ that is absolutely continuous with respect to the Lebesgue measure.  Indeed, for example, for each $n\geq1$,  we can partition the interval $[0,1]$ into intervals of length $\alpha_i^n$, $1\leq i\leq n$, and take a point $x_i^n$ in the $i$-th interval. Then, for every continuous function $f:[0,1] \to \mathbb{R}$ we have that $\int_0^1 f\mathrm{d}Q_n=\sum_{i=1}^n\alpha_i^n f(x_i^n)\to\int_0^1f(x)\mathrm{d}x=\int_0^1 f\mathrm{d}P$, hence $Q^n\to P$ in the weak-$*$-topology. Being sums of Dirac-measures it is clear that $\alpha(Q_n)=-1$ and thus
	\[
	-1=\liminf_{n \to \infty} \alpha(Q_n) \leq\alpha(Q)=0,
	\]
	whence $\alpha$ is not lower semicontinuous.
	
We now claim that  there exists $X\in \widehat{C}^{P}\setminus C^P$.
	Indeed, define $X_{k}=k$ $P$-a.s. with some fixed number $k\in(0,1)$. Then $X_{k}\in \widehat{C}^{P}$. Note that
	\begin{equation*}
		\widehat{C}^{P}=\{X\in L^{\infty}(P) : \EE_{Q}[-X]\le0,\ \forall
		Q\in\mathcal{P}^{P}\}
	\end{equation*}
	and, trivially, for each $Q\ll P$:
	$$\EE_Q[-X_{k}]=-kQ(\Omega)=-k<0.$$
	However $X_{k}\notin C^{P}$.
	
	Indeed, we cannot find a version $\widetilde{X}=k$ $P$-a.s. such that
	\begin{equation*}
		\int_{0}^{1} \widetilde{X}(\omega)\delta_{x}(\mathrm{d}\omega)\ge1, \qquad
		\forall x\in[0,1].
	\end{equation*}
	Here, $\delta_{x}$ denotes the Dirac distribution centered in
	$x\in[0,1]$, which lies for all $x \in \cM_{1}\setminus\mathcal{P}^{P}$. Therefore, for any $\widetilde{X}$ with $P(X_k=\widetilde{X})=1$
	\[
	\sup_{Q \in \cM_{1}} \EE_Q[-\widetilde{X}]-\alpha(Q) \geq -k+1 > 0
	\]
and thus clearly also for the infimum. This proves the claim and shows together with Lemma \ref{lem:relation} that the minimax inequality is strict.
	\end{example}
	
The above example shows that lower semicontinuity of the penalty function $\alpha$ is crucial. By replacing it by the corresponding $\alpha_{\min}$, which is here $\alpha_{\min} \equiv -1$ 
(compare with the shifted worst case risk measure as defined in Section \ref{sec:worstcase} for $a=1$),
the previous gap $\rho^P (X) > \widehat{\rho}^P(X)=\widetilde{\rho}^P(X)$  disappears, which is a consequence of Remark \ref{rem:minalpha}.

\subsection{Convex law invariant risk measures: worst case properties and comparison of the two approaches} \label{sec:comparisionlawinv}

We now specialize the setting further and consider
 \emph{law invariant risk measures.}
 To define law-invariance we need to start with a reference measure $\P$ right from the beginning, which could correspond to some `worst case probability measure' used by the regulator.
 
 Recall that  a risk measure $\rho$ is called law-invariant if
\[
\rho(X) = \rho (Y), \quad \text{if } X=Y \,\ \P\text{-a.s.}
\]

For a convex law invariant risk measure (that satisfies automatically the Fatou property, see \cite{JST:06}), any penalty function is of the following form
\begin{align}\label{eq:alphacanonical1}
\alpha(Q)= \begin{cases}
			\widetilde{\alpha}(Q) \quad	&\text{for} \ Q\in\mathcal{Q}^{\P} \subseteq \mathcal{P}^{\P}, \\
			+\infty \quad &\text{for} \ Q\in\mathcal{M}_{1}\setminus \mathcal{Q}^{\P} ,
		\end{cases}
\end{align}
where $\mathcal{Q}^{\P}$ denotes some weak-$*$-closed subset of $\mathcal{P}^{\P}$ and $\widetilde{\alpha}(Q)$ satisfies
\[
\widetilde{\alpha}(Q) < \infty, \quad \forall Q \in \mathcal{Q}^{\P}
\]
and
\[
\inf_{Q \in \mathcal{Q}^{\P}} \widetilde{\alpha}(Q) > -\infty.
\]
This follows from \cite[Lemma 4.30 and Theorem 4.31]{FS:2011} and the Fatou property.

Instead of an arbitrary penalty function, we here fix $\P$ and consider penalty functions $\alpha$  of the more specific form \eqref{eq:alphacanonical1}. 
Hence, throughout this section $\rho$ denotes the (pointwise defined) convex risk measure as of \eqref{eq:convexriskm} now with $\alpha$ given by \eqref{eq:alphacanonical1} with the fixed $\P$.
Its acceptance set is given by
\begin{align}\label{eq:Clawinv}
C=\{X \in \mathcal{X} :  \sup_{Q \in \mathcal{Q}^{\P}} \{\E_Q[-X]-\widetilde{\alpha}(Q)\} \leq 0 \}.
\end{align}

The goal is now to analyze  $\rho^P$ and $\widehat{\rho}^P$ defined with respect to some new measure $P$.  
Here, according to the first approach, $\rho^P$ is defined via
\[
\rho^P(X)=\inf_{\{\widetilde{X} : P(\widetilde{X}=X)=1\}}\rho(\widetilde{X})= \inf_{\{\widetilde{X} : P(\widetilde{X}=X)=1\}} \sup_{Q \in \mathcal{Q}^{\P}}\{\mathbb{E}_Q[-\widetilde{X}]-\widetilde{\alpha}(Q)\}, \quad X \in L^{\infty}(P),
\]
while $\widehat{\rho}^P$ is given by \eqref{eq:rhohat} or equivalently by \eqref{rhohat} with
 penalty function 
\begin{align}\label{eq:alphaR}
\alpha^{P}(Q)=\begin{cases}
	\widetilde{\alpha}(Q) \quad	&\text{for} \
	Q\in\mathcal{Q}^{{\P}}\cap\mathcal{P}^{P}, \\
	+\infty \quad &\text{for} \ Q\in\mathcal{M}_{1}\setminus({Q}^{{\P}}\cap\mathcal{P}^{P}).
	\end{cases}
\end{align}
The corresponding acceptance set  $C^P$ and $\widehat{C}^P$ are defined as in Section \ref{fixrisk} and Section \ref{fixpenalty}.

\begin{remark}\label{rem:third_approach}
In the current setup we could introduce a third approach where we consider the sets $\mathcal{Q}^{\P}$ as a certain function of the measures and define an alternative penalty function
$\overline{\alpha}^P$ via
\begin{align*}
\overline{\alpha}^{P}(Q)=\begin{cases}
	\widetilde{\alpha}(Q) \quad	&\text{for} \
	Q\in\mathcal{Q}^{P}, \\
	+\infty \quad &\text{for} \ Q\in\mathcal{M}_{1}\setminus\mathcal{Q}^{P}
	\end{cases}
\end{align*}
so that $\overline{\alpha}^{\P}$ is equal to \eqref{eq:alphacanonical1}. To give an example, in the case of Expected Shortfall (see Section \ref{sec:ES}) the sets $\mathcal{Q}^{P}$ are given by
\[
\mathcal{Q}^{P}=\{Q\in\cM_{1} : Q\ll P, \ \text{and}\ \frac{\mathrm{d}Q}{\mathrm{d}P}\le \lambda^{-1}, \P\text{-a.s.}\}
\]
for some $\lambda \in (0,1)$.
Note that in general $\mathcal{Q}^P \neq \mathcal{Q}^{\P} \cap \mathcal{P}^P$. However, when $P \ll \P$, there are  certain situations where $\mathcal{Q}^P = \mathcal{Q}^{\P} \cap \mathcal{P}^P$ holds. Indeed, the superhedging price (see Section \ref{ref:superhedging}) is one example for this equality. In the case of the Expected Shortfall it however does not hold true.

As already stated in Remark \ref{absconsP},
a crucial property of our second approach is that $\mathcal{Q}^{\P} \cap \mathcal{P}^{P'} \subseteq \mathcal{Q}^{\P} \cap \mathcal{P}^{P}$ whenever $P' \ll P$ so that we get the corresponding ordering $\widehat{\rho}^{P'} \leq \widehat{\rho}^{P}$, which would not necessarily hold true if we pursued the above outlined third approach. 
\end{remark}

Our first result of this section relates $\rho^P$ and $\widehat{\rho}^P$ with the starting risk measure $\rho$ and shows  the following order relation
\begin{align}\label{eq:ordering}
\widehat{\rho}^P|_{\mathcal{X}} \leq \rho^P|_{\mathcal{X}} \leq \rho=\rho^{\P}|_{\mathcal{X}}= \widehat{\rho}^{\P}|_{\mathcal{X}}
\end{align}
for \emph{all} $P \in \mathcal{M}_1$ (and not only absolutely continuous ones). This justifies the motivation given in the introduction that the reference measure $\P$ should correspond to some `worst case probability measure' capturing all the risk considered to be relevant since $\rho=\rho^{\P}|_{\mathcal{X}}= \widehat{\rho}^{\P}|_{\mathcal{X}}$ always yields the maximal risk. We emphasize again that also in cases where ${\P}$ is not  dominating the measure $P$ the risk computed with respect to $P$ is \emph{always} smaller or equal to the risk with respect to $\P$, which is a consequence of the definitions of $\rho^P$ and $\widehat{\rho}^P$ (see also Remark \ref{remark:orderetal} (iii) below).

Concerning the comparison of  $\rho^P$ and $\widehat{\rho}^P$, 
 we obtain different results depending on the relation between $\P$ and $P$.
If $P$ is either dominating or singular with respect to $\P$, then we always have $\rho^P=\widehat{\rho}^P$, which is also  asserted in the next proposition. 
The proposition also shows  whenever $P \perp \P $, the risk of both $\rho^P$ and  $\widehat{\rho}^P$ is  $-\infty$. Moreover, if $\P \ll P$, then  $\rho^P|_{\mathcal{X}}=\widehat{\rho}^P|_{\mathcal{X}}$  is equal to $\rho =\rho^{\P}|_{\mathcal{X}}=\widehat{\rho}^{\P}|_{\mathcal{X}}$, which is of course in line with \eqref{eq:ordering}.

\begin{proposition} \label{prop:PvsR}
Let $\P \in \mathcal{M}_1$ be a reference measure and let $\alpha$ be of form \eqref{eq:alphacanonical1} with $\mathcal{Q}^{\P}$ some weak-$*$-closed subset of $\mathcal{P}^{\P}$. 
Then for all $P \in  \mathcal{M}_1$ which satisfy either ${\P} \ll P$ or $\P \perp P$, it holds that $C^P= \widehat{C}^P$ and thus  $\rho^P=\widehat{\rho}^P=\widetilde{\rho}^P$. Hence, the minimax equality holds true
\begin{align*}
\widetilde{\rho}^P(X)&=\sup_{Q \in \mathcal{Q}^{\P}}\inf_{\{\widetilde{X} \in \mathcal{X} : P(\widetilde{X}=X)=1\}} \{\mathbb{E}_Q[-\widetilde{X}]-\widetilde{\alpha}(Q)\}\\
&= \inf_{\{\widetilde{X} \in \mathcal{X} : P(\widetilde{X}=X)=1\}} \sup_{Q \in \mathcal{Q}^{\P}}\{\mathbb{E}_Q[-\widetilde{X}]-\widetilde{\alpha}(Q)\} =\rho^P(X).
\end{align*}
 In particular, if ${\P} \ll P$,  the restriction of  $\rho^P=\widehat{\rho}^P=\widetilde{\rho}^P$ to $\mathcal{X}$ is equal to $\rho$. If ${\P}\perp P$, then $\rho^P=\widehat{\rho}^P= \widetilde{\rho}^P \equiv -\infty$.\\
 Moreover, for all $P \in \mathcal{M}_1$ we have
 \[
 \widehat{\rho}^P|_{\mathcal{X}} \leq \rho^P|_{\mathcal{X}} \leq \rho = \rho^{\P}|_{\mathcal{X}} =\widehat{\rho}^{\P}|_{\mathcal{X}}.
 \]
 and  also 
 \[
 \widehat{\rho}^P \leq \rho^P = \rho^{\P} =\widehat{\rho}^{\P}, \quad \text{on} \quad  L^{\infty}(P) \cap L^{\infty}({\P}).
 \]
\end{proposition}

\begin{remark}\label{remark:orderetal}
\begin{enumerate}
    \item 
Observe that in contrast to the minimax result as of Proposition \ref{th:minimax} and  Corollary \ref{th:minimax1} below, here we do not need to assume any properties on $\alpha$ or $\Omega$. 

\item Note that Proposition \ref{prop:PvsR} implies that  $\rho^{\P}=\widehat{\rho}^{\P}$ \emph{always} holds if $\alpha$ is of form \eqref{eq:alphacanonical1}.

\item Notice that for general $P$, we do not necessarily have
$\rho^P=\widehat{\rho}^P$. For $\widehat{\rho}^P$ the penalty function $\alpha^P$ is given by \eqref{eq:alphaR} and  $\mathcal{Q}^{\P}\cap\mathcal{P}^{P}=\mathcal{Q}^{\P} \cap\mathcal{P}^{P_1}$ where $P_1$ denotes the normalized absolutely continuous part of the Lebesgue decomposition of $P$ with respect to $\P$. Then it is easy to see that $\widehat{\rho}^P=\widehat{\rho}^{P_1}$. Similarly we have $\rho^P=\rho^{P_1}$ as shown in the proof below. But we do not necessarily have $\rho^{P_1}=\widehat{\rho}^{P_1}$ (compare also Proposition \ref{lemma_dual_regulator}).

\end{enumerate}

\end{remark}

\begin{proof}

We start by relating $\rho^P $ with $\rho$ for the cases ${\P} \ll P$ and ${\P} \perp P$.
Recall that $C^P$ is defined by
$$C^P=\{ X \in L^{\infty}(P)\, : \inf_{\{\widetilde{X} \in \mathcal{X} : P(\widetilde{X}=X)=1\}} \rho(\widetilde{X}) \leq 0\},$$
which we now compare with $C$ given in \eqref{eq:Clawinv}.

If  ${\P} \ll P$, we have $C=C^P \cap \mathcal{X}$. Indeed, it always holds that $C \subseteq C^P \cap \mathcal{X}$. To prove $C^P \cap \mathcal{X} \subseteq C$, let $X \in C^P \cap \mathcal{X}$. Since all $P$-nullsets are ${\P}$-nullsets and in turn also $Q$-nullsets for all $Q \in \mathcal{Q}^{\P}$, we have for all $Q \in \mathcal{Q}^{\P}$ and for \emph{all} $\widetilde{X}$ satisfying $P(\widetilde{X}=X)=1$, that
$\E_Q[-\widetilde{X}]=\E_Q[-X]$ and thus in turn
$$
\rho(\widetilde{X})=\sup_{Q\in \mathcal{Q}^{\P}}\{\E_Q[-X]-\widetilde{\alpha}(Q)\}=\rho(X).
$$

Since $\rho(\widetilde{X})$  assumes the same value for every $\widetilde{X}$, we have $\rho^P(X)=\rho(\widetilde{X}) $ which is $\leq 0$ as $X \in C^P \cap \mathcal{X}$ . Thus $\rho(X) \leq 0$ as well, whence $X\in C$. This proves the claim and implies that for all ${\P} \ll P$,  $\rho^P|_{\mathcal{X}}=\rho^{\P}|_{\mathcal{X}}=\rho$ and similarly $\rho^P=\rho^{\P}|_{L^{\infty}(P)}$. 

  In the case  ${\P} \perp P$, $C \subset C^P=L^{\infty}(P)$. Indeed, since  ${\P} \perp P$, there exists a set $B$ with $P(B)=1 $ and $\P(B^c)=1$. Let $X \in L^{\infty}(P)$ and choose $\widetilde{X} \in \mathcal{X}$ such $X=\widetilde{X}$ on $B$ and
$\widetilde{X} \geq -\inf\widetilde{\alpha}(Q) $
 on $B^c$. For example, take, for each $n\geq1$,  $\widetilde{X}^n=X\ind_B+n\ind_{B^c}$. 
 As $-\inf\widetilde{\alpha}(Q)<\infty$ we have that, for all $n$ large enough, $n\geq-\inf\widetilde{\alpha}(Q)$.
By taking the infimum over all $n$, it follows that
 $$\rho^P(X) = \inf_{\{\widetilde{X} \, | \, P(\widetilde{X}=X)=1 \}} \sup_{Q \in \mathcal{Q}^{\P}}\{ \E_Q[-\widetilde{X}]- \widetilde{\alpha}(Q)\}\leq \inf_{n\geq 1}(-n)-\inf_{Q \in \mathcal{Q}^{\P}}\widetilde{\alpha}(Q) \equiv -\infty.$$

 Next we analyze the relation between $\widehat{\rho}^P $ and $\rho$ in the cases ${\P} \ll P$
 and ${\P} \perp P$, which in turn will allow us to relate also $\rho^P$ with $\widehat{\rho}^P$.  For $\widehat{\rho}^P$, the penalty function $\alpha^P$ is given by \eqref{eq:alphaR}.

If ${\P}\ll P$, then  $\mathcal{Q}^{\P}\cap\mathcal{P}^{P}=\mathcal{Q}^{\P}$, thus $\alpha^P=\alpha$ and $\widehat{\rho}^P|_{\mathcal{X}}=\rho$ and $\widehat{\rho}^P=\rho^{\P}|_{L^{\infty}(P)}$. Since  $\rho^P=\rho^{\P}|_{L^{\infty}(P)}$ as proved above, we also get $\widehat{\rho}^P= \rho^P$.

If ${\P} \perp P$, then $\mathcal{Q}^{\P}\cap\mathcal{P}^{P}=\emptyset$, yielding $\alpha^{P}(Q)=+\infty$ for all $Q\in\cM_{1}$ and in turn $\widehat{\rho}^P\equiv-\infty$ such that $\widehat{C}^P=L^{\infty}(P)$. This already proves the equality  $C^P= \widehat{C}^P$ and $\rho^P=\widehat{\rho}^P$ for ${\P} \perp P$.
The assertion concerning the minimax equality follows from Lemma \ref{lem:relation}.

 Finally, consider an arbitrary $P \in \mathcal{M}_1$. Then by Lebesgue's decomposition theorem we can decompose $P$ into $P=\widetilde{P}_1 + \widetilde{P}_2$ with $\widetilde{P}_1 \ll P $ and $\widetilde{P}_2 \perp P$. We denote by $P_1$ and  by $P_2$ the corresponding normalized probability measures. As above there exists some set $B$ such that $P_2(B)=1$ and $\P(B)=P_1(B)=0$. Hence $\P(B^c)=1$ and $P_1(B^c)=1$. 
 Thus, for all $Q \in \mathcal{Q}^{\P}$ and for all $\widetilde{X}$ with $P(\widetilde{X}=X)=1$, we have $\E_Q[-\widetilde{X}]=\E_Q[-\widetilde{X} \ind_{B^c}]$ since $\E_Q[-\widetilde{X} \ind_{B}]=0$ due to the singularity between $P_2$ and $Q$.
 We therefore get
 \[
 \rho(\widetilde{X})=\sup_{Q\in \mathcal{Q}^{\P}}\{\E_Q[-\widetilde{X}]-\widetilde{\alpha}(Q)\}=\rho(\widetilde{X}\ind_{B^c}),
 \]
 whence $\rho^P=\rho^{P_1}$ on $L^{\infty}(P)$ since we only need that $\widetilde{X}=X$, $P_1$-almost surely.
 As $P_1 \ll \P$, we thus have $\rho^P \leq \rho^{\P}$ on $L^{\infty}(P) \cap L^{\infty}(\P)$.
 This proves the last assertion since $\widehat{\rho}^P \leq \rho^P$ holds by Lemma \ref{lem:relation}.
\end{proof}

As indicated in Remark \ref{remark:orderetal} (iii)
the relation between $\rho^P$ and $\widehat{\rho}^P$ is more subtle for general $P \in \mathcal{M}_1$. 
We consider in the next proposition the case when $P$  is absolutely continuous with respect to the reference measure $\P$ and derive sufficient conditions when equality holds true. By  Remark \ref{remark:orderetal} (iii) this then translates also to the general case.

\begin{proposition}\label{lemma_dual_regulator}
Let $\P \in \mathcal{M}_1$ be a reference measure and let $\alpha$ be of form \eqref{eq:alphacanonical1} with $\mathcal{Q}^{\P}$ some weak-$*$-closed subset of $\mathcal{P}^{\P}$. 
Let $P \ll \P$ and consider for $Q \in  \mathcal{Q}^{\P} \setminus (\mathcal{Q}^{\P} \cap \mathcal{P}^P) $  the Lebesgue decomposition into  $Q=\widetilde{Q}_1+\widetilde{Q}_2$ where $\widetilde{Q}_1,\widetilde{Q}_2$ are nonnegative measures such that $\widetilde{Q}_1 \ll P$ and $\widetilde{Q}_2 \perp P$.
If

\begin{enumerate}
\item
there exists a set $B \in \mathcal{F}$ such that
$P(B)=1$ and $\widetilde{Q}_2(B)=0$ for the singular parts $\widetilde{Q}_2$ of all $Q \in \mathcal{Q}^{\P} \setminus (\mathcal{Q}^{\P} \cap \mathcal{P}^P)$,

\item $\inf_{\mathcal{Q}^{\P}} \widetilde{\alpha}(Q)= \inf_{\mathcal{Q}^{\P} \cap \mathcal{P}^P} \widetilde{\alpha}(Q) $,  
\end{enumerate}
then
$C^P= \widehat{C}^P $ and thus
$
\rho^P = \widehat{\rho}^P= \widetilde{\rho}^P$. Hence,
the minimax equality
\begin{align*}
\widetilde{\rho}^P(X)&=\sup_{Q \in \mathcal{Q}^{\P}}\inf_{\{\widetilde{X} \in \mathcal{X} : P(\widetilde{X}=X)=1\}} \{\mathbb{E}_Q[-\widetilde{X}]-\widetilde{\alpha}(Q)\}\\
&= \inf_{\{\widetilde{X}   \in \mathcal{X}: P(\widetilde{X}=X)=1\}} \sup_{Q \in \mathcal{Q}^{\P}}\{\mathbb{E}_Q[-\widetilde{X}]-\widetilde{\alpha}(Q)\} =\rho^P(X),
\end{align*}
holds true.
\end{proposition}

\begin{remark}\label{remark_countable}
Concerning Condition (i) of Proposition \ref{lemma_dual_regulator}, recall that $\widetilde{Q}_2 \perp P$ means that there are disjoint sets $A(\widetilde{Q}_2)$ and $B(\widetilde{Q}_2)$ such that $A(\widetilde{Q}_2) \cup B(\widetilde{Q}_2)=\Omega$ and $P(A(\widetilde{Q}_2))=0$ and $\widetilde{Q}_2(B(\widetilde{Q}_2))=0$.
If $\bigcup_{\widetilde{Q}_2} A(\widetilde{Q}_2)$ can be described by a countable union, i.e.
\[
\bigcup_{\widetilde{Q}_2} A(\widetilde{Q}_2)= \bigcup_{i=1}^{\infty} A(\widetilde{Q}_2^i),
\]
then $P(\bigcup_{\widetilde{Q}_2} A(\widetilde{Q}_2))=0$ and the set $B$ with $P(B)=1$ can be chosen as $B=\bigcap_{i=1}^{\infty} B(\widetilde{Q}_2^i)$. This clearly always works when $\Omega$ is countable and thus the Condition (i) is in this case always satisfied.

Observe additionally that Condition (ii) is satisfied for any coherent risk measure where $\widetilde{\alpha}(Q)$ is chosen to be $0$ for all $Q\in\mathcal{Q}^{\P}$.
\end{remark}

\begin{proof}
As shown in Lemma \ref{lem:relation}, $\rho^P \geq \widehat{\rho}^P$, hence we only have to prove the converse inequality.
As $\alpha^{P}(Q)$ is given by \eqref{eq:alphaR}, we have
$$
\widehat{\rho}^P(X)=\sup_{Q \in \mathcal{Q}^{\P} \cap\mathcal{P}^{P}} \{\E_Q[-X]-\widetilde{\alpha}(Q)\},
$$
while 
$\rho^P$ is given by
\[
\rho^P(X)=
\inf_{\{\widetilde{X} \in \mathcal{X} : P(\widetilde{X}=X)=1\}}\sup_{Q \in \mathcal{Q}^{\P}} \{\E_Q[-\widetilde{X}]-\widetilde{\alpha}(Q)\}.
\]
Consider a measure
 $Q \in \mathcal{Q}^{\P} \setminus (\mathcal{Q}^{\P}\cap\mathcal{P}^{P})$ and denote by ${Q}_i$ its normalized Lebesgue decomposition, i.e. $Q_i= \frac{\widetilde{Q}_i}{\widetilde{Q}_i (\Omega)}$ for $i=1,2$. Then due to the assumption that there exists a  set $B$ with $P(B)=1$ and  $\widetilde{Q}_2(B)=0$ for the singular parts of all $Q \in \mathcal{Q}^{\P} \setminus (\mathcal{Q}^{\P}\cap\mathcal{P}^{P})$, we can consider $\widetilde{X}=X$ on $B$ and $\widetilde{X} =n \in \mathbb{R}$  on $B^c$. This yields
\begin{align*}
\sup_{Q \in\mathcal{Q}^{\P}} \{\E_Q[-\widetilde{X}]-\widetilde{\alpha}(Q)\} &= \sup_{Q \in\mathcal{Q}^{\P}}\{\widetilde{Q}_1 (\Omega)\E_{{Q}_1}[-\widetilde{X}1_B]+ \widetilde{Q}_2 (\Omega)\E_{{Q}_2}[-\widetilde{X}1_{B^c}]-\widetilde{\alpha}(Q)\}\\
&\leq  \sup_{\widetilde{Q}_2(\Omega) \in [0,1]}\big\{ (1-\widetilde{Q}_2 (\Omega)) \sup_{Q \in\mathcal{Q}^{\P} \cap \mathcal{P}^P}\E_{Q}[-\widetilde{X}1_B]\\
&\quad - \widetilde{Q}_2 (\Omega)n -\inf_{Q \in \mathcal{Q}^{\P}}\widetilde{\alpha}(Q)\big \}.
\end{align*}
We now optimize the last expression over $\widetilde{Q}_2(\Omega)$, which leads to
\[
\widetilde{Q}_2(\Omega)=\begin{cases}
0 & \quad  \text{if } n \geq -\sup_{Q \in \mathcal{Q}^{\P} \cap \mathcal{P}^P} \E_{Q}[-X],\\
1 & \quad \text{else}.
\end{cases}
\]
and
\begin{align*}
&\sup_{Q \in\mathcal{Q}^{\P}} \{\E_Q[-\widetilde{X}]-\widetilde{\alpha}(Q)\} \\
&\quad \leq \begin{cases}
\sup_{Q \in\mathcal{Q}^{\P} \cap \mathcal{P}^P} \E_{Q}[-X]-\inf_{Q \in \mathcal{Q}^{\P}} \widetilde{\alpha}(Q) & \quad  \text{if } n \geq -\sup_{Q \in \mathcal{Q}^{\P} \cap \mathcal{P}^P} \E_{Q}[-X],\\
-n -\inf_{Q \in \mathcal{Q}^{\P}} \widetilde{\alpha}(Q) & \quad \text{else}.
\end{cases}
\end{align*}
Taking the infimum over $n$ on both sides (recall that $\widetilde{X}=n$ on $B^c$) then yields,
\begin{align*}
\rho^P(X)&=\inf_{\{\widetilde{X} \in \mathcal{X} : P(\widetilde{X}=X)=1\}} \sup_{Q \in\mathcal{Q}^{\P}} \{\E_Q[-\widetilde{X}]-\widetilde{\alpha}(Q)\} \leq
\inf_{\{\widetilde{X} =X 1_B+ n 1_{B^c}\}} \sup_{Q \in\mathcal{Q}^{\P}} \{\E_Q[-\widetilde{X}]-\widetilde{\alpha}(Q)\} \\
&\leq \sup_{Q \in\mathcal{Q}^{\P} \cap \mathcal{P}^P} \E_{Q}[-X]-\inf_{Q \in \mathcal{Q}^{\P}} \widetilde{\alpha}(Q)=\sup_{Q \in\mathcal{Q}^{\P} \cap \mathcal{P}^P} \{\E_{Q}[-X]- \widetilde{\alpha}(Q)\}\\&= \widehat{\rho}^P(X)
\end{align*}
where the second last equality follows from Assumption (ii) and shows $\rho^P \leq \widehat{\rho}^P$. The assertion concerning the minimax equality also follows from this lemma.
\end{proof}
Alternatively to the above propositions   one can apply again  Sion's minimax result to determine conditions for the equality of $\rho^P$ and $\widehat{\rho}^P$.
In this case fewer  assumptions are needed on the set $\mathcal{Q}^{\P}$ but more on the penalty function $\widetilde{\alpha}$ and $\Omega$. Note that we do not need to distinguish between the different relations to the initial measure $\P$.

\begin{corollary}\label{th:minimax1}
Suppose that $\Omega$ is compact. Let $\P \in \mathcal{M}_1$ be a reference measure and let $\alpha$ be of form \eqref{eq:alphacanonical1} with $\mathcal{Q}^{\P}$ some weak-$*$-closed convex subset of $\mathcal{P}^{\P}$.
If  $Q \mapsto \widetilde{\alpha}(Q)$ is lower semicontinuous  (with respect to the weak-$*$-topology) and convex,
then for any $P \in \mathcal{M}_1$
the minimax equality
\begin{align*}
\widehat{\rho}^P(X)=\widetilde{\rho}^P(X)&=\sup_{Q \in \mathcal{Q}^{\P}}\inf_{\{\widetilde{X} \in \mathcal{X} : P(\widetilde{X}=X)=1\}} \{\mathbb{E}_Q[-\widetilde{X}]-\widetilde{\alpha}(Q)\}\\
&= \inf_{\{\widetilde{X}   \in \mathcal{X}: P(\widetilde{X}=X)=1\}} \sup_{Q \in \mathcal{Q}^{\P}}\{\mathbb{E}_Q[-\widetilde{X}]-\widetilde{\alpha}(Q)\} =\rho^P(X),
\end{align*}
holds true.
\end{corollary}

\begin{proof}
The proof is analogous to Proposition \ref{th:minimax} by replacing $\alpha$ by $\widetilde{\alpha}$ and noting that the weak-$*$-closed set $\mathcal{Q}^{\P} \subseteq \mathcal{M}_1$ is compact, since $\Omega$ is compact.

\end{proof}

\begin{remark}
Similarly as in Remark \ref{rem:minalpha} all conditions of Corollary \ref{th:minimax1}  are automatically satisfied, 
if $\Omega$ is compact and $\alpha$ happens to be $\alpha_{\min}$.
This is because $\alpha_{\min}$ is convex and weak-$*$-lower semicontinuous (see Theorem 2.3.1 in \cite{Z:02}), which translates to $\widetilde{\alpha}$ then.  
\end{remark}

\section{Examples}\label{sec:examples}

In the following we shall give well-known examples of risk measures to illustrate the above comparison results.

 We here additionally discuss possible differences between
starting  with a fixed acceptance set or rather with a fixed penalty function to define the initial risk measure.

\subsection{Superhedging price as risk measure}\label{ref:superhedging}

As above let $(\Omega,\mathcal{F})$ be fixed and $\P$ a reference measure.

Let $S=(S_{t})_{t\ge0}$ denote the  price process of an asset and let $T$ be a finite time horizon. We denote by $\mathcal{H}$ the set of admissible trading strategies (see e.g. \cite{FS:2011}) and by $(H\bullet S)_{T}$ the stochastic integral of $H\in\mathcal{H}$ with respect to $S$, representing the gains and losses up to time $T$, obtained from trading in $S$ according to the strategy $H$. We recall here the standard definition of the superhedging price.

\begin{definition}[Superhedging price]\label{superhedg_def}
	Let $\P \in \cM_{1} $ and $X\in L^{\infty}(\P)$ be a European contingent claim. Then the superhedging price $\overline{\pi}(X)\in\RR$ of $X$ is defined as
	\begin{equation}\label{eq:superhedge}
		\overline{\pi}(X)=\inf\{x\in\RR : \exists H\in\mathcal{H} \ \text{with}\ X\le x+ (H\bullet S)_{T}, \, \P\text{-a.s.}\}.
	\end{equation}
\end{definition}

Notice that by the superhedging duality (see \cite{DS:1994}, \cite{DS:1999}), $\overline{\pi}(X)$ can be equivalently written as
\begin{equation*}
	\overline{\pi}(X)=\sup_{Q\in\cM_{a}(\P)}\EE_{Q}[X],
\end{equation*}
where $\cM_{a}(\P)$ denotes the set of absolutely continuous separating measures for $S$ with respect to the reference measure $\P$, see for instance \cite{DS:1994, K:1997, DS:1999}. This already gives the dual representation (with the corresponding $\alpha_{min}$).
 
By the previous definition, for a financial agent
``superhedging" means to find a self-financing strategy with minimal initial capital
which covers any possible future obligation resulting from selling
a European contingent claim. \\
 
We now show how the superhedging price can be embedded in our frameworks. There are two possibilities: one which does not exploit convexity and starts just with a monetary risk measure coming from an acceptance set to introduce $\rho_1^{\P}$, fully in spirit of Section \ref{fixrisk}; and the second which uses convexity and the minimal penalty function, in spirit of Section \ref{fixpenalty} to define in turn $\rho_2^{\P}$ and $\widehat{\rho}_2^{\P}$. We shall show that they all coincide.

We start by defining an acceptance set via
\begin{equation*}
	C:=\{X\in\mathcal{X} :\exists H\in\mathcal{H} \
	\text{with} \ (H\bullet S)_{T}(\omega)\ge X(\omega), \ \forall\omega\in\Omega\},
\end{equation*}
where the stochastic integral is here understood in discrete time, so that it is also well-defined in a pointwise sense.
This then induces the corresponding risk measure
$\rho_1:\mathcal{X}\to\RR$
via
\begin{equation}\label{eq:suphedge1}
	\rho_1(X)=\inf\{m\in\RR : m+X\in C\}.
\end{equation}
For the reference measure $\P$, we then introduce according to Definition~\ref{rhoP}
\[
\rho^{\P}_1(X)=\inf_{\{\widetilde{X} \in \mathcal{X} : \P(\widetilde{X}=X)=1\}}\rho_1(\widetilde{X}),
\]
which is equal to \eqref{eq:superhedge} since there the inequality has to hold only $\P$-a.s.\\

 Since the superhedging price is a coherent risk measure, and thus convex we can introduce it via its dual representation. Fix $\alpha$ in  \eqref{eq:convexriskm} to be the following (minimal) penalty function
 \begin{align}\label{eq:alphasuperhedge}
 \alpha(Q)= 	\begin{cases}
		0 \quad	&\text{for} \ Q\in\cM_{a}(\P), \\
		+\infty \quad &\text{for} \ Q\in\mathcal{M}_{1}\setminus \cM_{a}(\P),
	\end{cases}
	\end{align}
with $\cM_{a}(\P)$ denoting again the set of absolutely continuous
separating measures with respect to $\P$, which correspond in the setup of discrete time to absolutely continuous
martingale measures with respect to $\P$. 
Pursuing the second approach we can define a  risk measure $\rho_2$ (pointwise) via \eqref{eq:convexriskm} using the penalty function \eqref{eq:alphasuperhedge}.
Then $\alpha^{\P}$ as of \eqref{alphaP} is equal to $\alpha$ and thus

	\begin{align*}
		\widehat{\rho}_2^{\P}(X)&=\sup_{Q\in\cM_{1}}\{\mathbb{E}_{Q}[-X]-\alpha^{\P}(Q)\}\\
		&=\sup_{Q\in\cM_{a}(\P)}\mathbb{E}_{Q}[-X].
	\end{align*}
Note that we could have started with any function $\alpha$ such that $\alpha^{\P}$ as defined in \eqref{alphaP}  is equal to the right hand side of \eqref{eq:alphasuperhedge} to get the same expression for $\widehat{\rho}_2^{\P}$. In particular the values outside $\mathcal{P}^{\P}$ do not matter. Using $\rho_2$, we can of course also consider
$\rho_2^{\P}$ defined via
\begin{equation*}
	\rho_2^{\P}(X)=\inf_{\{\widetilde{X} \in \mathcal{X} : \P(\widetilde{X}=X)=1\}}\rho_2(\widetilde{X})=\inf_{\{\widetilde{X}   \in \mathcal{X}: \P(\widetilde{X}=X)=1\}}\sup_{Q \in \cM_{1}}\{\mathbb{E}[-\widetilde{X}]-\alpha(Q)\}.
\end{equation*}
Due to the current choice of $\alpha$, we can then verify that $\rho^{\P}_1=\rho_2^{\P}=\widehat{\rho}_2^{\P}$.
Indeed, since $\alpha$ is of form \eqref{eq:alphacanonical1}, $\rho_2^{\P}=\widehat{\rho}_2^{\P}$ follows from Proposition \ref{prop:PvsR} by setting $P=\P$,  and $\rho^{\P}_1=\rho_2^{\P}$ by checking that the respective acceptance sets are equal.\\

In the following we introduce an additional measure $P\ll \P$ and
consider the setup of Proposition \ref{lemma_dual_regulator} as well as Corollary \ref{th:minimax1}  for the superhedging price in a simple trinomial model.

\begin{example}\label{ex:trinomial}

 Let us consider an arbitrage-free one-period trinomial model $S=(S_t)_{t=0,1}$, with zero interest rate, under some reference measure $\P$ which assigns positive probability to each scenario $\omega_i$ for $i=1,2,3$. Here, $\omega_1$ corresponds to the up-scenario, $\omega_2 $ to the middle one and $\omega_3$ to the down-scenario. We suppose that $\mathcal{F}_0$ is trivial.
Let $P\ll \P$ be such that $P(\{\omega_1, \omega_3\})=1$ and $P(\{\omega_2\})=0$.
 We consider as risk measure $\rho$ again the superhedging price and introduce it via \eqref{eq:convexriskm} with $\alpha$ of form \eqref{eq:alphasuperhedge}.
 
 Note that in this case the conditions of Proposition \ref{lemma_dual_regulator} hold true. Indeed,  the first condition is clearly satisfied by  Remark \ref{remark_countable} as $\lvert \Omega \lvert=3$. The second condition is also fulfilled since $\widetilde{\alpha}\equiv 0$ on $\mathcal{Q}^{\P}=\cM_{a}(\P)$. The assumption of Corollary \ref{th:minimax1} on $\widetilde{\alpha}$ is thus clearly also satisfied.

Consider now the concrete example
$S_0=2$, $S_1(\omega_1)=4$, $S_1(\omega_2)=3$ and $S_1(\omega_3)=1$. Then for all $Q \in \mathcal{M}_a(\P)$, we have $Q(\{\omega_1\})\in [0,\frac{1}{3}]$,
$Q(\{\omega_2\})= \frac{1- 3 Q(\{\omega_1\})}{2}$ and
$Q(\{\omega_3\})=\frac{1 + Q(\{\omega_1\})}{2}$. For $Q(\{\omega_1\})=\frac{1}{3}$ it holds that $Q \in \mathcal{M}_a({\P}) \cap \mathcal{P}^P$ and this is actually the only element therein. Take now some claim
$X \in L^{\infty}(P)$.
Then since $\alpha^P$ is given by \eqref{eq:alphaR}
\[
\widehat{\rho}^P(X)=-\frac{1}{3} X(\omega_1) -\frac{2}{3} X(\omega_3),
\]
while
\begin{align*}
\rho^P(X)&=\inf_{X(\omega_2)\in \mathbb{R}} \sup_{Q(\{\omega_1\}) \in [0, \frac{1}{3}]}(- Q(\{\omega_1\}) X(\omega_1)- \frac{1- 3 Q(\{\omega_1\})}{2} X(\omega_2)  - \frac{1+ Q(\{\omega_1\})}{2}X(\omega_3)) \\
&=\inf_{X(\omega_2)\in \mathbb{R}} \sup_{Q(\{\omega_1\}) \in [0, \frac{1}{3}]}
\left(-\frac{1}{2}(X(\omega_2)+X(\omega_3))+  Q(\{\omega_1\})(-X(\omega_1)+ \frac{3}{2} X(\omega_2)-  \frac{1}{2} X(\omega_3))\right)\\
&=\inf_{X(\omega_2)\in \mathbb{R}}\begin{cases}
-\frac{1}{2}(X(\omega_2)+X(\omega_3)) & \quad \text{if } X(\omega_2) \leq \frac{2}{3}X(\omega_1) + \frac{1}{3} X(\omega_3),\\
-\frac{1}{3} X(\omega_1) -\frac{2}{3} X(\omega_3) & \quad \text{else.}
\end{cases}\\
&= -\frac{1}{3} X(\omega_1) -\frac{2}{3} X(\omega_3),
\end{align*}
showing that we have indeed equality.

\end{example}

In the following we provide an example showing that Assumption (ii) of Lemma \ref{lemma_dual_regulator} or the lower-semicontinuity assumption of Theorem \ref{th:minimax1} on $\widetilde{\alpha}$ is crucial to have equality.

\begin{example}
Consider exactly the same trinomial model as  in Example \ref{ex:trinomial} above. Define however the  penalty function $\alpha:\cM_{1}\to\RR\cup\{+\infty\}$ as follows
	\begin{align*}
	\alpha(Q)&=
	\begin{cases}
		\frac{Q(\{\omega_2\})^2}{2}-1 \quad	&\text{for} \ Q\in\cM_{a}(\P) \text{ s.t. } Q(\{\omega_2\}) >0, \\
		0 \quad	&\text{for} \ Q\in\cM_{a}(\P) \text{ s.t. } Q(\{\omega_2\}) =0,  \\
		+\infty \quad &\text{for} \ Q\in\mathcal{M}_{1}\setminus \cM_{a}(\P).
	\end{cases}
\end{align*}
Note that with this choice $\widehat{\rho}^P$ still corresponds to the superhedging price, which is in this case the price of the perfect replication. To show that $\rho^P$ is different, consider the trivial claim $X \in L^{\infty}(P)$, i.e. $X(\omega_1)=0$ and $X(\omega_3)=0$.
Then $\widehat{\rho}^P= 0$ while
\begin{align*}
\rho^P(X)&=\inf_{X(\omega_2)\in \mathbb{R}} \sup_{Q(\{\omega_2\}) \in [0, \frac{1}{2}]}\left(- Q(\{\omega_2\}) X(\omega_2) - \frac{Q(\{\omega_2\})^2}{2}+1\right) =1,
\end{align*}
showing that Condition (ii) of Proposition~\ref{lemma_dual_regulator}  is in general needed for equality. Note that $ Q(\{\omega_2\}) \mapsto \alpha(Q(\{\omega_2\}))$ fails to be lower semi-continuous at $0$, whence the Condition of Corollary \ref{th:minimax1} are also not satisfied.

\end{example}

\subsection{Expected Shortfall}\label{sec:ES}

As next example we consider Expected Shortfall  (ES) of level $\lambda\in(0,1)$ for some reference measure $\P$, another example of a coherent risk measure.  Fix its (minimal) penalty function
\begin{equation}\label{eq:alphaES1}
		\alpha(Q)=\begin{cases}
			0 \quad	&\text{for} \
			Q\in\mathcal{Q}^{\P}, \\
			+\infty \quad &\text{for} \ Q\in\mathcal{M}_{1}\setminus
			\mathcal{Q}^{\P},
		\end{cases}
\end{equation}
where $\mathcal{Q}^{\P}:=\{Q\in\cM_{1} : Q\ll \P, \ \text{and}\ \frac{\mathrm{d}Q}{\mathrm{d}\P}\le \lambda^{-1}, \P\text{-a.s.}\}$. Then we define $\rho$ via \eqref{eq:convexriskm}, i.e. $\rho(X)= \sup_{Q\in \mathcal{M}_1}\{\mathbb{E}_{Q}[-X]- \alpha(Q)\}$ for $X \in \mathcal{X}$
and  since $\alpha=\alpha^{\P}$, $\widehat{\rho}^{\P}$ is given by  
\begin{equation*}
	\widehat{\rho}^{\P}(X)=\sup_{Q\in \mathcal{Q}^{\P}}\mathbb{E}_{Q}[-X], \quad X \in L^{\infty}(\P).
\end{equation*}
Additionally, we introduce $\rho^{\P}$ according to the first approach as
\begin{equation*}
	\rho^{\P}(X)=\inf_{\{\widetilde{X}   \in \mathcal{X}: \P(\widetilde{X}=X)=1\}} \rho(X)=\inf_{\{\widetilde{X}   \in \mathcal{X}: \P(\widetilde{X}=X)=1\}}\sup_{Q\in\cM_{1}}\{\mathbb{E}_{Q}[-\widetilde{X}]-\alpha(Q)\}.
\end{equation*}

As $\alpha$ is of form \eqref{eq:alphacanonical1}, we can conclude by applying Proposition \ref{prop:PvsR} with $P=\P$ that $\rho^{\P}=\widehat{\rho}^{\P}$.

Also in the context of ES we now introduce a new measure $P \ll \P$ and consider
the setup of Lemma \ref{lemma_dual_regulator} and Corollary \ref{th:minimax1}  in a simple trinomial model, where by the same arguments as in Example \ref{ex:trinomial} all conditions are satisfied. We here also aim to show the equality of $\rho^P$ and $\widehat{\rho}^P$ by an explicit calculation.

\begin{example}

For some measure $P \ll \P $, let $\alpha^P $ be given by \eqref{eq:alphaR} with $\alpha$ of form \eqref{eq:alphaES1} and consider
\begin{equation*}
	\widehat{\rho}^{P}(X):=\sup_{Q\in \mathcal{M}_{1}}\{\EE_{Q}[-X]-\alpha^{P}(Q)\}.
\end{equation*}
Our goal is to show explicitly  equality with $$\rho^P(X)= \inf_{\{\widetilde{X}   \in \mathcal{X}: P(\widetilde{X}=X)=1\}}\sup_{Q\in\cM_{1}}\{\mathbb{E}_{Q}[-\widetilde{X}]-\alpha(Q)\}$$ in the concrete trinomial setup of Example \ref{ex:trinomial}.
Let $\P$ be such that $\P(\{\omega_{i}\})=\frac{1}{3}$ for all $i=1,2,3$  and $P$ as above, i.e.  $P(\{\omega_1,\omega_3\})=1$ and $P(\{\omega_2\})=0$.
Then the elements $Q\in  \mathcal{Q}^{\P} \cap\mathcal{P}^{P}$ are simply such that
\begin{equation*}
	Q(\{\omega_2\})=0, \qquad \frac{Q(\{\omega_i\})}{\P(\{\omega_i\})}\le\frac{1}{\lambda}, \ \text{for}\ i=1,3,
\end{equation*}
with the constraints $Q(\{\omega_{i}\})\ge0$ for $i=1,3$ and $Q(\{\omega_{1}\})+Q(\{\omega_{3}\})=1$. Notice that if we fix $\lambda=2/3$, then $\mathcal{Q}^{\P} \cap\mathcal{P}^{P}=\{Q^{\ast}\}$ with
\begin{equation*}
		Q^{\ast}(\{\omega_{1}\})=Q^{\ast}(\{\omega_{3}\})=\frac{1}{2},\qquad Q^{\ast}(\{\omega_{2}\})=0.
\end{equation*}
For a claim  $X\in L^{\infty}(P)$, we thus have
\[
\widehat{\rho}^P(X)=-\frac{1}{2}(X(\omega_1) +X (\omega_3))
\]
while
\begin{align*}
\rho^P(X)&=\inf_{X(\omega_2)\in \mathbb{P}}\ \sup_{Q \in \mathcal{Q}^{\P}} \{-Q(\{\omega_1\}) X(\omega_1) -Q(\{\omega_2\}) X(\omega_2)-Q(\{\omega_3\}) X(\omega_3)\}.
\end{align*}
To compute $\rho^{P}(X)$, let us consider first the inner part which can be seen as a linear programming problem. Observe that the objective functional is a linear function in $Q(\omega_{i})$ for all $i=1,2,3$ and that $X(\omega_{1})$ and $X(\omega_{3})$ are fixed. Additionally recall that the set of probability measures with $\Omega$ as above, can be identified with the two-dimensional simplex and that $\mathcal{Q}^{\P}=\{Q\in \cM_{1} : Q(\{\omega_{i}\})\le\frac{1}{2}, \ \forall i=1,2,3\}$. Therefore the supremum will be attained on one of the three extremal points of $\mathcal{Q}^{\P}$, namely in those probability measures under which two scenarios occur with probability $\frac{1}{2}$ and the third with probability zero. For these reasons 
\begin{equation*}
    \sup_{Q \in \mathcal{Q}^{\P}} - \sum_{j=1}^{3}Q(\{\omega_j\}) X(\omega_j)=-\frac{1}{2}\max\{X(\omega_{1})+X(\omega_{3}),  X(\omega_{1})+X(\omega_2), X(\omega_{3})+X(\omega_2)\}.
\end{equation*}
Let us now consider the outer problem with the infimum, hence $\rho^{P}(X)$ can be equivalently rewritten as
\begin{align*}
    \rho^{P}(X)&=-\frac{1}{2}\sup_{X(\omega_{2})\in\mathbb{R}} \min\{X(\omega_{1})+X(\omega_{3}),  X(\omega_{1})+X(\omega_2), X(\omega_{3})+X(\omega_2)\}.
\end{align*}
Let us consider without loss of generality two cases, namely $X(\omega_{1})<X(\omega_{3})$ and $X(\omega_{1})\ge X(\omega_{3})$. In the first one \begin{equation*} we have
    \min\{X(\omega_{1})+X(\omega_{3}),  X(\omega_{1})+X(\omega_2), X(\omega_{3})+X(\omega_2)\}=\begin{cases}
X(\omega_{1})+X(\omega_{3}) \ \ \text{if } X(\omega_3) \leq X(\omega_2),\\
X(\omega_{1})+X(\omega_{2}) \ \ \text{if } X(\omega_3) > X(\omega_2).
\end{cases}
\end{equation*}
Similarly in the second case, we get
\begin{equation*}
    \min\{X(\omega_{1})+X(\omega_{3}),  X(\omega_{1})+X(\omega_2), X(\omega_{3})+X(\omega_2)\}=\begin{cases}
X(\omega_{1})+X(\omega_{3}) \ \ \text{if } X(\omega_1) \leq X(\omega_2),\\
X(\omega_{2})+X(\omega_{3})  \ \ \text{if } X(\omega_1) > X(\omega_2).
\end{cases}
\end{equation*}
Therefore in both cases by taking the supremum over $X(\omega_{2})\in\mathbb{R}$ we have that $$\rho^P(X)=-\frac{1}{2}(X(\omega_1)+X(\omega_3)),$$
which coincides with $\widehat{\rho}^{P}(X)$ as claimed.

\end{example}

\subsection{The worst case risk measure}\label{sec:worstcase}
The worst case risk measure is an additional example of a risk measure which is part of our framework. Similarly as for the superhedging price it can be introduced via an acceptance set,  in spirit of Section \ref{fixrisk} or via a particular penalty function in spirit of  Section \ref{fixpenalty}.

In the following we fix $a\ge0$ to cover a shifted version of the original worst case risk measure, but we stress that for $a=0$ the definition coincides with the one considered in \cite{FS:2011}. Let $C_{a}$ be the convex acceptance set defined as
\begin{equation*}
	C_{a}:=\{X\in\mathcal{X}: X(\omega)\ge a, \ \forall \omega\in\Omega\}.
\end{equation*}
Then the corresponding risk measure is given by
\begin{equation}\label{eq:worst1}
	\rho_{\max,1}(X)=\inf\{m\in\RR : m+X\in C_{a}\}
\end{equation}
or equivalently by
\begin{equation*}
	\rho_{\max,1}(X)=-\inf_{\omega\in\Omega}X(\omega)+a.
\end{equation*}

If we fix the penalty function 
$\alpha(Q)= - a$ for all $Q \in \mathcal{M}_1$, we can define $$\rho_{max,2}(X):=\sup_{Q\in\cM_{1}}\{\mathbb{E}_{Q}[-X] -\alpha(Q)\}=\sup_{Q\in\cM_{1}}\mathbb{E}_{Q}[-X] + a$$ 
and see that this coincides with \eqref{eq:worst1}.

For a reference measure $\P$, we can then introduce according to Definition~\ref{rhoP} the acceptance set to be 
\begin{equation*}
	C^{\P}=\{X\in L^{\infty}(\P): X\ge a,\ \P\text{-a.s.}\}
\end{equation*}
and
\begin{equation}\label{wcrmeasure}
	\rho_{\max, 1}^{\P}(X):=\inf_{\{\widetilde{X} \in \mathcal{X} : \P(\widetilde{X}=X)=1\}}\rho_{\max,1}(\widetilde{X})=- \P\text{-}\essinf_{\omega \in \Omega} X(\omega)+a.
\end{equation}

Since $\rho_{\max,1}=\rho_{\max, 2}$, it clearly holds that
\begin{equation*}
	\rho_{\max,1}^{\P}(X)=\rho_{\max,2}^{\P}(X)=\inf_{\{\widetilde{X}   \in \mathcal{X}: \P(\widetilde{X}=X)=1\}}\sup_{Q\in \cM_{1}}\mathbb{E}[-\widetilde{X}]+a.
\end{equation*}

Starting from the second approach we can additionally introduce  for all $X\in L^{\infty}(\P)$,
\begin{equation*}
	\widehat{\rho}_{\max,2}^{\P}(X)=\sup_{Q\in \mathcal{P}^{\P}}\mathbb{E}[-X]+a.
\end{equation*}
This also coincides with $\rho_{\max,1}^{\P}(X)=\rho_{\max,2}^{\P}(X)$ which follows  from Proposition \ref{prop:PvsR} by setting $\widetilde{\alpha}=-a$ on $\mathcal{P}^{\P}$.\\

\section{Robust risk measures and their classical representation}\label{sec:robust}

We shall now use the above approaches of the defining risk measures with respect to certain sets of  reference measures to include model uncertainty. To this end assume that we
are given a family of probability measures $\cM \subset\cM_{1}$ defined on the same measurable space  $(\Omega, \mathcal{F})$. 
In order
to be able to consider all risk measures for all $P\in\cM$, we choose $X$ to be in the space $\mathcal{X}^{\cM}$, where
\begin{equation}\label{X_tilde}
	\mathcal{X}^{\cM}:=\bigcap_{P\in \cM}L^{\infty}(P).
\end{equation}
Note that if $X\in \mathcal{X}^{\cM}$ then, for all $P\in\cM$, we have that $P(X\in\mathcal{X})=1$.
Let us now define the robust risk measure
$\rho^{\cM}$ as the supremum over all risk measures $\rho^{P}$, with $P\in\cM$ and similarly $\widehat{\rho}^{\cM}$  as the supremum over all risk measures $\widehat{\rho}^{P}$.

\begin{definition}[Robust risk measures]\label{robust_rm}
	Let $X\in\mathcal{X}^{\cM}$. Define the robust risk measure of $X$
	according to the class of models $\cM$ as
	\begin{equation*}
		\rho^{\cM}(X):=\sup_{P\in \cM}\rho^{P}(X).
	\end{equation*}

	Similarly, let $X\in\mathcal{X}^{\cM}$ and fix $\alpha:\cM_{1}\to\RR\cup\{+\infty\}$. Define the robust risk measure of $X$ according to the class of model $\cM$ as
	\begin{align*}
		\widehat{\rho}^{\cM}(X):&=\sup_{P\in\cM}\widehat{\rho}^{P}(X)\\
		&=\sup_{Q\in \mathcal{M}_{1}}\{\EE_{Q}[-X]-\alpha^{\cM}(Q)\},
	\end{align*}
where
\begin{align*}
	\alpha^{\cM}(Q)=
	\begin{cases}
		\alpha(Q) \quad	&\text{for} \
		Q\in\mathcal{P}^{\cM}, \\
		+\infty \quad &\text{for} \ Q\in\mathcal{M}_{1}\setminus  \mathcal{P}^{\cM},
	\end{cases}
\end{align*}
with $\mathcal{P}^{\cM}:=\bigcup_{P\in \cM}\mathcal{P}^{P}$.

\end{definition}
\begin{remark}
Observe that the previous robust risk measures fall into the setup of \emph{generalized risk measures} introduced recently by \cite{FLW:2021}. Let $\cM_{1}$ be the set of probability measure here assumed to be atomless, and $2^{\cM_{1}}$ the collection of subsets of $\cM_{1}$.  A generalized risk measure is a map $\Psi:\mathcal{X}\times 2^{\cM_{1}}\to \mathbb{R}$. Consider the following axioms 
\begin{itemize}
    \item[(A1)] $\Psi(X\lvert\mathcal{Q})\le\Psi(X\lvert\mathcal{R})$, for all $X\in\mathcal{X}$, $\mathcal{Q}\subset \mathcal{R}\subset \cM_{1}$;
    \item[(A2)] $\Psi(X\lvert\{P\})\le\Psi(Y\lvert\{P\})\ \forall P\in\mathcal{Q} \implies \Psi(X\lvert\mathcal{Q})\le\Psi(Y\lvert\mathcal{Q})$, for any $X,Y\in\mathcal{X}$;
    \item[(A3)]$\Psi(X\lvert\mathcal{Q})\le \sup_{P\in\mathcal{Q}}\Psi(X\lvert\{P\})$, for all $X\in\mathcal{X}$ and $\mathcal{Q}\subset\cM_{1}$.
\end{itemize}
In \cite{FLW:2021} is shown that if $\Psi$ satisfies A1-A3, then $\Psi(X\lvert \mathcal{Q})=\sup_{P\in\mathcal{Q}}\Psi(X\lvert\{P\})$. By setting $\mathcal{Q}=\cM$ and $\Psi(X\lvert\{P\}):=\rho^{P}(X)$ in our setup we notice immediately that 
$$\Psi(X\lvert \mathcal{M})=\sup_{P\in\cM}\Psi(X\lvert\{P\})=\sup_{P\in\cM}\rho^{P}(X),$$
can be seen as generalized risk measure on $\mathcal{X}\times 2^{\cM_{1}}$. Similar considerations hold for $\widehat{\rho}^{P}$.\\
In Section \ref{sec:classical}, under certain assumptions on $\cM$, we will address the problem of finding a measure $\PP$ such that 
\begin{equation*}
    \Psi(X\lvert \mathcal{M})=\Psi(X\lvert\{\PP\}), \qquad \forall X\in L^{\infty}(\PP),
\end{equation*}
i.e. we ask when the above generalized risk measure has a representation as a classical risk measure.
\end{remark}

Let us now introduce an extended set of probability measures induced by $\cM$, namely 
\begin{equation}\label{cloconvM}
	c\mathcal{M}:=\overline{conv}(\cM),
\end{equation}
where   $\overline{conv}$ denotes the set of all probability measures which are countable convex combinations of the elements in $\cM$.

Sometimes we shall replace $\mathcal{M}$ via the set $c\mathcal{M}$ in the definition of the robust risk measure $\rho^{\cM}$, i.e.~
$$\rho^{c\cM}(X)=
\sup_{P\in c\cM}\rho^{P}(X), \qquad \forall X\in
\mathcal{X}^{c\cM}$$ and similarly for $\widehat{\rho}^{c\cM}$.
Observe that  $\rho^{{\cM}}(X)\leq \rho^{c\cM}(X)$ for $X\in
\mathcal{X}^{c\cM}$ with a possible strict inequality for certain $X \in \mathcal{X}^{c\cM}$ and the analogous statement holds for $\widehat{\rho}^{c\cM}$.

In the following we give  an example where the inequality is strict even with finite convex combinations. We consider in particular the case of $\Omega$ being compact and we choose the risk measure to be coherent. We work with $\widehat{\rho}^{P}$ but remark  that the next example also holds true for $\rho^{P}$, which follows from Corollary \ref{th:minimax1} or Proposition \ref{prop:PvsR} with $R=P$.

\begin{example}
Let $(\Omega,\mathcal{F})=([0,1],\mathcal{B}([0,1]))$ and suppose $\cM$ is induced by a parameter set $\Theta$ with a $\sigma$-algebra $\mathcal{D}$, i.e.,
	 $(\Theta,\mathcal{D})=([0,1],\mathcal{B}([0,1]))$ as well.
	 Let $N>2$ and $\gamma=\frac1{N}$. We define $\cM:=\{P^{\theta} : \theta\in\Theta\}$ where $P^{\theta}$ is given as follows. For all $A\in\mathcal{B}([0,1])$
	\begin{equation*}
		P^{\theta}(A):=\begin{cases}\frac2{\gamma} \Lambda(A\cap [\theta,\theta+\frac{\gamma}2])&\text{ for all $\theta\in[0,1-\frac{\gamma}2]$},\\
\frac2{\gamma}\Lambda(A\cap[\theta,1])&\text{ for all $\theta\in(1-\frac{\gamma}{2},1]$},\end{cases}
\end{equation*}
	where $\Lambda$ denotes the Lebesgue measure on $[0,1]$. We set in the following $\widehat{\rho}^{\theta}:=\widehat{\rho}^{P^{\theta}}$ for all $\theta\in\Theta$ and $\widehat{\rho}^{\cM}$ as above. Additionally with $VaR_{\gamma}^{\Lambda}$ we refer to the Value at Risk of level $\gamma\in (0,1)$ under $\Lambda$.\\

	Let $X:=-\ind_{[0,2\gamma]}$, then
	\begin{equation*}
		VaR_{\gamma}^{\Lambda}(X)=\inf\{m\in\RR : \Lambda(X+m<0)<\gamma\}=1.
	\end{equation*}
	Indeed for all $\varepsilon>0$,
	\begin{equation*}
		\Lambda(-\ind_{[0,2\gamma]}+1-\varepsilon<0)=\Lambda(-\varepsilon\ind_{[0,2\gamma]}+(1-\varepsilon)\ind_{(2\gamma,1]}<0)=\Lambda([0,2\gamma])=2\gamma>\gamma,
	\end{equation*}
	and
	\begin{equation*}
		\Lambda(-\ind_{[0,2\gamma]}+1<0)=0<\gamma.
	\end{equation*}
	Moreover, for all $\beta\in(0,\gamma]$ we have that $VaR_{\beta}^{\Lambda}(X)=1$, thus $\rho(X)=1$ where $\rho$ is the Expected Shortfall of level $\gamma$ as defined in \ref{sec:ES} via the set $\mathcal{Q}^{\Lambda}$. Define similarly as in Section \ref{sec:ES}
	\begin{equation*}
		\widehat{\rho}^{\theta}(X)=\sup_{Q\in \mathcal{M}_{1}}\{\EE_{Q}[-X]-\alpha^{\theta}(Q)\},
	\end{equation*}
	where
	$$\alpha^{\theta}(Q)=\begin{cases}
		0 \quad	&\text{for} \
		Q\in\mathcal{Q}^{\Lambda}\cap\mathcal{P}^{\theta}, \\
		+\infty \quad &\text{for} \ Q\in\mathcal{M}_{1}\setminus({Q}^{\Lambda}\cap\mathcal{P}^{\theta}),
	\end{cases}$$
with
\begin{equation*}
	\mathcal{Q}^{\Lambda}=\{Q\in\cM_{1}: Q\ll \Lambda,\ \text{and}\ \frac{\mathrm{d}Q}{\mathrm{d}\Lambda}\le\gamma^{-1}, \Lambda-a.s.\}.
\end{equation*}
Let us consider now $Q\in \mathcal{Q}^{\Lambda}\cap\mathcal{P}^{\theta}$. First of all since $Q\in \mathcal{Q}^{\Lambda}$ we notice that
\begin{align*}
	Q\left([\theta,\theta+\frac{\gamma}{2}]\right)&=\EE_{\Lambda}\left[\frac{\mathrm{d}Q}{\mathrm{d}\Lambda}\ind_{[\theta,\theta+\gamma/2]}\right]\\
	&\le \frac{1}{\gamma}\cdot\frac{\gamma}{2}=\frac{1}{2}.
\end{align*}
On the other hand since $Q\in \mathcal{P}^{\theta}$, we have that $\operatorname{supp}(Q)\subseteq[\theta,\theta+\frac{\gamma}{2}]$ yielding a contradiction as
\begin{equation*}
	Q\left([\theta,\theta+\frac{\gamma}{2}]\right)=1.
\end{equation*}
Hence $ \mathcal{Q}^{\Lambda}\cap\mathcal{P}^{\theta}=\emptyset$ for all $\theta\in[0,1]$, meaning that the penalty function always attains infinite values, namely $\alpha^{\theta}(Q)=+\infty$ for all $Q\in\cM_{1}$ and all $\theta\in[0,1]$. Therefore $\widehat{\rho}^{\theta}(X)=-\infty$ and consequently $\widehat{\rho}^{\cM}(X)=-\infty$.\\
Let us now consider $c\cM=\overline{conv}(\cM)$ and show that $\Lambda\in c\cM$. To this extent fix the sequence $(\tilde{\theta}_{k})_{k=0}^{2N}\subseteq[0,1]$ such that for all $k=0,\dots,2N$ we define $\tilde{\theta}_{k}:=k\frac{\gamma}{2}=\frac{k}{2N}$. Then for every $A\in\mathcal{B}([0,1])$
\begin{equation*}
	\Lambda(A)=\frac{\gamma}{2}\sum_{k=1}^{2N}\frac{2}{\gamma}\Lambda\left(A\cap \left[\tilde{\theta}_{k},\tilde{\theta}_{k}+\frac{\gamma}{2}\right]\right)=\sum_{k=1}^{2N}\frac{\gamma}{2}P^{\tilde{\theta}_{k}}(A),
\end{equation*}
hence $\Lambda\in c\cM$. In particular $\widehat{\rho}^{c\cM}(X)=\widehat{\rho}^{\Lambda}(X)=\rho(X)=1>\widehat{\rho}^{{\cM}}(X)$.
\end{example}

\subsection{Classical representation}\label{sec:classical}

Let us now come to one of our main goals namely to show equality of the robust risk measure (with respect to a set $\cM$) with a classical risk measure with respect to one single probability measure $\mathbb{P}$.
In this section we analyze under which conditions this is the case.

Let us start by recalling a version of the well known result of
\cite{HS:1949} (see Lemma~7).

\begin{theorem}[\cite{HS:1949}]\label{hs}
	Let $(\Omega,\mathcal{F},\mathbb{P})$ be a probability space. Let $\mathcal{M}$
	be a set of $\mathbb{P}$-absolutely continuous probability measures on
	$(\Omega,\mathcal{F})$ that is closed under countable convex combinations.
	Suppose that, for each $A\in\mathcal{F}$ with $\mathbb{P}(A)>0$, there exists
	$P\in\mathcal{M}$  (depending on $A$) with $P(A)>0$. Then there exists
	$P^0\in\mathcal{M}$ such that, for all $A\in\mathcal{F}$ with
	$\mathbb{P}(A)>0$, we have that $P^0(A)>0$, that is $P^0$ and $\mathbb{P}$ are
	equivalent probability measures.
\end{theorem}

Next we introduce the notion of \emph{generalized equivalence},   a version of absolute continuity between the elements of $\cM$ and a given probability measure which does not necessarily  belong to $\cM$. It means in particular that there exists a dominating measure for the set $\mathcal{M}$.

\begin{definition}[Generalized equivalence]\label{abscon}
	Let $\cM$ be a family of probability measures defined on $(\Omega,\mathcal{F})$ and let $\PP$ be a some
	probability measure defined on $(\Omega,\mathcal{F})$. The model $\cM$ and the measure $\PP$ are said to be \emph{generalized equivalent} 
	\begin{enumerate}
	    \item if
$P\ll\PP$ for all
	$P\in\cM$, denoted by $\cM\ll \PP$, and 
	\item if, for $A\in\mathcal{F}$, the condition 
 $P(A)=0$ for all $P\in\cM$
	implies $\PP(A)=0$, denoted by
	$\PP\ll\cM$.
		\end{enumerate}
\end{definition}

Using Theorem \ref{hs} and the notion of generalized equivalence now allows  to formulate the precise statement when $\rho^{\PP}(X)=\rho^{\cM}(X)$ holds.

\begin{theorem}\label{th:nonhat}
Let $\cM \subseteq \mathcal{M}_1$ and  $\mathbb{P}$ a probability measure on $(\Omega, \mathcal{F})$. Suppose $\mathbb{P}$ is generalized equivalent to $\cM$ and that $\cM$ is closed under countable convex combinations. Then for all $X\in L^{\infty}(\PP)$, 
\begin{equation*}
    \rho^{\PP}(X)=\rho^{\cM}(X).
\end{equation*}
\end{theorem}
\begin{proof}
By Theorem \ref{hs} there exists $P^{0}\in\cM$ such that   $\mathbb{P}$ is equivalent to $P^{0}$. Since in particular $\mathbb{P}$ is dominated by $P^{0}$, by Remark \ref{abs:fixset} we have that for all $X\in L^{\infty}(\PP)$
\begin{equation*}
    \rho^{\mathbb{P}}(X)= \rho^{P^{0}}(X)\le \rho^{\cM}(X).
\end{equation*}
On the other hand since for all $P\in\cM$, $P\ll \mathbb{P}$ then, for all $X\in L^{\infty}(\PP)$
\begin{equation*}
    \rho^{\cM}(X)=\sup_{P\in\cM}\rho^{P}(X)\le\rho^{\mathbb{P}}(X),
\end{equation*}
which concludes the proof.
\end{proof}

\begin{remark}
Note that the above result also holds true when $\rho^{\mathbb{P}}$ and $\rho^{\mathcal{M}}$ are replaced by $\widehat{\rho}^{\mathbb{P}}$ and $\widehat{\rho}^{\mathcal{M}}$, simply by applying the same proof. As it also follows from the next more general result, we only formulated it for $\rho ^{\mathbb{P}}$ and $\rho^{\mathcal{M}}$.
\end{remark}

In the following theorem we characterize when equality    of $\rho^{\PP}(X)$ with $\rho^{\cM}(X)$ holds.

\begin{theorem}\label{prop:criteria}
Let $\cM \subseteq \mathcal{M}_1$ and $\mathbb{P}$ a probability measure on $(\Omega, \mathcal{F})$. Fix a penalty function $\alpha: \mathcal{M}_1 \to \mathbb{R} \cup \{+\infty\}$.
\begin{itemize}
\item
Then, the following assertions are equivalent.
\begin{itemize}
\item[(i)] For all $X \in \mathcal{X}^{\mathcal{M}} \cap L^{\infty}(\PP)$, it holds $\widehat{\rho}^{\cM}(X)=\widehat{\rho}^{\mathbb{P}}(X)$.
\item[(ii)]
\[
\bigcup_{P \in \cM} \mathcal{P}^P \cap \{Q : \alpha(Q) < \infty\}=\mathcal{P}^{\mathbb{P}} \cap \{Q : \alpha(Q) < \infty\}.
\]
\end{itemize}
\item The above assertions are implied by the following condition:
\begin{itemize}
\item[(iii)]
All $P \in \cM$ are dominated by $\mathbb{P}$ and there is a measure $P_0 \in \cM$ such that $\mathbb{P} \ll P_0$.
\end{itemize}
Moreover, if $\alpha$ is finitely valued on $\bigcup_{P \in \cM} \mathcal{P}^{P} \cup \mathcal{P}^{\P}$, (iii) is equivalent to (i) and (ii).
\item Condition (iii) is implied if $\cM$ is closed under countable convex combinations and $\mathbb{P}$ generalized equivalent to $\cM$.
In this case  $\mathcal{X}^{\mathcal{M}} \cap L^{\infty}(\PP)=L^{\infty}(\PP)$.
\end{itemize}
\end{theorem}

\begin{proof}
	$(i)\implies (ii)$. Let $X\in  \mathcal{X}^{\mathcal{M}} \cap L^{\infty}(\PP)$. Consider
	\begin{equation*}
		\widehat{\rho}^{\cM}(X)=\sup_{P \in \cM}\sup_{Q\in\cM_{1}}\{\EE_{Q}[-X]-\alpha^{P}(Q)\}=\sup_{Q\in \mathcal{M}_{1}}\{\EE_{Q}[-X]-\alpha^{\cM}(Q)\},
	\end{equation*}
	where
	\begin{equation}\label{alpha_tag}
		\alpha^{\cM}(Q)=\begin{cases}
			\inf_{P\in\cM}\alpha^{P}(Q) \quad	&\text{for} \
			Q\in\bigcup_{P \in \cM}\mathcal{P}^{P}, \\
			+\infty \quad &\text{for} \ Q\in\mathcal{M}_{1}\setminus \bigcup_{P \in \cM}\mathcal{P}^{P},
		\end{cases}
	\end{equation}
and
\begin{equation*}
	\widehat{\rho}^{\PP}(X)=\sup_{Q\in \mathcal{M}_{1}}\{\EE_{Q}[-X]-\alpha^{\PP}(Q)\},
\end{equation*}
with
	\begin{equation}\label{alpha_P}
	\alpha^{\PP}(Q)=\begin{cases}
		\alpha(Q) \quad	&\text{for} \
		Q\in\mathcal{P}^{\P}, \\
		+\infty \quad &\text{for} \ Q\in\mathcal{M}_{1}\setminus \mathcal{P}^{\P}.
	\end{cases}
\end{equation}

If $(i)$ holds true then,
\begin{equation*}
	\{X\in \mathcal{X}^{\mathcal{M}} \cap L^{\infty}(\PP): \widehat{\rho}^{\PP}(X)\le 0\}=\{X\in \mathcal{X}^{\mathcal{M}} \cap L^{\infty}(\PP): \widehat{\rho}^{\cM}(X)\le 0\}.
\end{equation*}
By the definition of the risk measures the previous yields
\begin{align*}
	&\{X\in \mathcal{X}^{\mathcal{M}} \cap L^{\infty}(\PP): \EE_{Q}[-X]\le\alpha^{\PP}(Q), \forall Q\in\cM_{1}\}\\
	&\quad=\{X\in \mathcal{X}^{\mathcal{M}} \cap L^{\infty}(\PP): \EE_{Q}[-X]\le\alpha^{\cM}(Q), \forall Q \in \cM_{1}\}.
\end{align*}
Hence, by the definition of the penalty functions \eqref{alpha_tag} and \eqref{alpha_P} we obtain
\begin{align*}
	&\{X\in  \mathcal{X}^{\mathcal{M}} \cap L^{\infty}(\PP): \EE_{Q}[-X]\le\alpha(Q),\ \forall Q\in\mathcal{P}_{f}^{\PP}\}\\
	&\quad =\{X\in  \mathcal{X}^{\mathcal{M}} \cap L^{\infty}(\PP): \EE_{Q}[-X]\le\inf_{P\in\cM}\alpha^{P}(Q),\ \forall Q \in \bigcup_{P \in \cM}\mathcal{P}_{f}^{P}\}\\
	&\quad =\{X\in \mathcal{X}^{\mathcal{M}} \cap L^{\infty}(\PP): \EE_{Q}[-X]\le\alpha(Q),\ \forall Q \in \bigcup_{P \in \cM}\mathcal{P}_{f}^{P}\},
\end{align*}
where $\mathcal{P}_{f}^{P}:=\mathcal{P}^{P}\cap \{Q\in\cM_{1} : \alpha(Q)<+\infty\}$ for any $P\in\cM_{1}$. Hence we can conclude that $\mathcal{P}_{f}^{\P}=\bigcup_{P \in \cM}\mathcal{P}_{f}^{P}$.\\
$(ii)\implies(i).$ Let $X\in \mathcal{X}^{\mathcal{M}} \cap L^{\infty}(\PP)$. Then,
\begin{align*}
	\widehat{\rho}^{\cM}(X)=\sup_{P \in \cM}\widehat{\rho}^{P}(X)&=\sup_{Q\in\cM_{1}}\{\EE_{Q}[-X]-\alpha^{\cM}(Q)\}\\
	&=\sup_{Q\in\bigcup_{P \in \cM}\mathcal{P}^{P}}\{\EE_{Q}[-X]-\inf_{P\in\cM}\alpha^{P}(Q)\}\\
	&=\sup_{Q\in\bigcup_{P \in \cM}\mathcal{P}_{f}^{P}}\{\EE_{Q}[-X]-\alpha(Q)\}\\
	&\overset{(\ast)}{=}\sup_{Q\in\mathcal{P}_{f}^{\PP}}\{\EE_{Q}[-X]-\alpha(Q)\}\\
	&=\widehat{\rho}^{\PP}(X),
\end{align*}
where in $(\ast)$ the hypothesis $(ii)$ is employed, yielding the identity between the risk measures. Observe that if $(iii)$ holds, this implies
\begin{equation*}
	\bigcup_{P \in \cM}\mathcal{P}^{P}=\mathcal{P}^{\PP}
\end{equation*}
 and thus $(ii)$ as the previous trivially gives $\bigcup_{P \in \cM}\mathcal{P}_{f}^{P}=\mathcal{P}_{f}^{\PP}$. Notice additionally that if $\cM$ is closed under countable convex combination and $\PP$ is generalized equivalent to $\cM$, then $(iii)$ is satisfied as a consequence of Halmos-Savage's Theorem.
\end{proof}

In the following we apply the above result to the situation considered in Section \ref{sec:comparisionlawinv}, which shows again that the reference measure with respect to which we consider law-invariance qualifies as `worst case probability measure'. Note that for this result the generalized equivalence property is not necessary and it thus works in a  non-dominated situation.

\begin{corollary} \label{cor:worstcaseprob}
Let $\cM \subseteq \mathcal{M}_1$ and let $\mathbb{P} \in \mathcal{M}$ a probability measure on $(\Omega, \mathcal{F})$ with respect to which we consider law-invariance as in Section \ref{sec:comparisionlawinv}. Fix a penalty function $\alpha: \mathcal{M}_1 \to \mathbb{R} \cup \{+\infty\}$ as in \eqref{eq:alphacanonical1} with respect to $\mathbb{P}$. Then for all $X \in  \mathcal{X}^{\mathcal{M}} \cap L^{\infty}(\PP)$ we have $$\widehat{\rho}^{\cM}(X)=\widehat{\rho}^{\mathbb{P}}(X)=\rho^{\mathbb{P}}(X).$$
\end{corollary}

\begin{proof}
We apply Theorem \ref{prop:criteria} and verify (ii). Note that
\[
\bigcup_{P \in \cM} \mathcal{P}^P \cap \{Q : \alpha(Q) < \infty\}= (\bigcup_{P \in \cM}  \mathcal{P}^P) \cap \mathcal{Q}^{\mathbb{P}} = \mathcal{Q}^{\mathbb{P}},
\]
where the latter follows from the fact that $\mathbb{P} \in \mathcal{M}$. 
Moreover, we clearly also have
\[
\mathcal{P}^{\mathbb{P}} \cap \{Q : \alpha(Q) < \infty\} = \mathcal{Q}^{\mathbb{P}},
\]
whence property (ii) is satisfied and the claim follows. The equality $\widehat{\rho}^{\mathbb{P}}(X)=\rho^{\mathbb{P}}(X)$ for all $X\in \mathcal{X}^{\mathcal{M}}\cap L^{\infty}(\P)$ follows from Proposition \ref{prop:PvsR}.
\end{proof}

Note that Corollary \ref{cor:worstcaseprob} could of course also have been deduced from Proposition \ref{prop:PvsR}. This result yields additionally the same assertion for $\rho^{\mathcal{M}}$ which we state here for completeness. Note again that the generalized equivalence property is not required.

\begin{corollary}  \label{cor:worstcaseprob1}
Let $\cM \subseteq \mathcal{M}_1$ and let $\mathbb{P} \in \mathcal{M}$ a probability measure on $(\Omega, \mathcal{F})$ with respect to which we consider law-invariance as in Section \ref{sec:comparisionlawinv}. Fix a penalty function $\alpha: \mathcal{M}_1 \to \mathbb{R} \cup \{+\infty\}$ as in \eqref{eq:alphacanonical1} with respect to $\mathbb{P}$. Then for all $X \in \mathcal{X}^{\mathcal{M}} \cap L^{\infty}(\PP)$ we have
	$$
		\widehat{\rho}^\mathbb{\mathcal{M}}(X)=\widehat{\rho}^\mathbb{P}(X)= \rho^\mathbb{P}(X)=\rho^\mathbb{\mathcal{M}}(X).
		$$
\end{corollary}

\begin{proof}
This is a consequence of the last statement of  Proposition \ref{prop:PvsR}.
\end{proof}

\begin{remark}
	\begin{enumerate}
		\item Let us stress the significance of the above results. Whenever $\mathbb{P}$ and $\mathcal{M}$ are generalized equivalent and $\mathcal{M}$ is closed under countable convex combinations, then the robust risk measure for a given $X\in
		L^{\infty}(\P)$ is nothing else than the classical risk measure for the
		probability $\P$. So the current robust theory, i.e., the theory under model
		uncertainty, just
		reduces to the classical theory. A particular case of this situation is
		if $\cM\subset\cM_{1}$ is countable, see Example~\ref{count} below.
		\item In the setup of Section \ref{sec:comparisionlawinv}  the probability measure with respect to which we consider law-invariance  is the natural candidate to achieve a classical representation \emph{without domination property and convexity assumptions} provided that it is contained in the set $\mathcal{M}$.

		\item Recall the space of random variables defined in \eqref{X_tilde}.
		Concerning the assumption $X\in L^{\infty}(\P)$, this cannot be
		replaced without additional considerations by $X\in\mathcal{X}^{\mathcal{M}}$, as we see in
		Example~\ref{tXversusP} below.
		This example shows that $X\in \mathcal{X}^{\mathcal{M}}$ does not imply $X\in L^{\infty}(\P)$
		not even in the case where $\cM$ is countable. It only holds for $\cM$
		finitely generated.
		
	\end{enumerate}
\end{remark}

\begin{example}\label{count}  Let $\Theta$ be a countable parameter space i.e., $\Theta=\{\theta_i, i\in\mathbb{N}\}$ and $\cM^{\Theta}:=\{P^{\theta} : \theta \in \Theta\}$ a collection of probability measures. Consider as the
	dominating measure
	$\P$ any countable convex combination that gives strictly positive weight to	all $\theta_i$, e.g.
	\begin{equation*}
		\P(A):=\sum_{i=1}^{+\infty}2^{-i}P^{\theta_i}(A), \qquad \forall
		A\in\mathcal{F}.
	\end{equation*}
	Then for all $i\in\NN$ we have that $P^{\theta_i}\ll \P$. Furthermore if $P^{\theta_i}(A)=0$, for all $i\in\NN$, then $\P(A)=0$, which implies $\P\ll\mathcal{M}^{\Theta}$.
\end{example}

\begin{example}\label{tXversusP}
	Let $(\Omega,\mathcal{F})=([0,1],\mathcal{B}([0,1]))$ and let $\cM=\{P^n:
	n\in\mathbb{N}\cup\{0\}\}$ where
	\begin{equation*}
		\frac{\mathrm{d}P^n}{\mathrm{d}\Lambda}(\omega)=2^{n+1}\ind_{(2^{-(n+1)},
			2^{-n}]}(\omega),
	\end{equation*}
	for all $\omega\in[0,1]$ where $\Lambda$ here denotes the Lebesgue measure. Define
	\begin{equation*}
		\P(A):=\sum_{n=0}^{+\infty}2^{-(n+1)}P^n(A), \qquad \forall
		A\in\mathcal{B}[0,1].
	\end{equation*}
	Then $\frac{\mathrm{d}\P}{\mathrm{d}\Lambda}=\ind_{[0,1]}$
	$\Lambda$--a.s., hence $\P$ can be identified with the Lebesgue-measure.
	Let
	$X(\omega)=\frac1{\omega}$. We have that, for all $n$, $P^n(X\leq
	2^{n+1})=1$, hence $X\in \mathcal{X}^{\mathcal{M}}=\bigcap_{n\ge0}L^{+\infty}(P^n)$. But
	on the other hand $X\notin
	L^{\infty}(\P)$.
\end{example}
\section{Mixture probability measure}\label{sec:mixture}

We now 
include model uncertainty assuming that we
are given a family of probability measures $\mathcal{M}^{\Theta}= \{P^\theta: \theta \in \Theta\}$ defined on the same measurable space  $(\Omega, \mathcal{F})$ and where $\Theta$
is a parameter space. We assume that, for all $A\in\mathcal{F}$,  $\theta\mapsto
P^{\theta}(A)$ is a
measurable map from $(\Theta, \mathcal{D})\to([0,1], \mathcal{B}([0,1]))$, where
$\mathcal{B}([0,1])$ are the Borel sets on $[0,1]$ and $\mathcal{D}$ is a
$\sigma$-algebra on
$\Theta$. We assume that
we have a priori subjective beliefs in the degree of validity of the possible
measures $P^{\theta}$, $\theta\in\Theta$. This means we are given a probability measure $ \nu $ on $ \Theta $. We understand uncertainty here from a Bayesian
viewpoint, i.e., we average different measures with respect to a given prior
distribution $\nu$.
\subsection{Definition and properties}

\begin{definition}[Mixture probability measure]\label{mixturemeas}
	Let $\P$ on $(\Omega,\mathcal{F})$ be the mixture probability measure defined as
	$$\P(A):=\int_{\Theta}P^{\theta}(A)\nu(\mathrm{d}\theta), \qquad \forall A \in \mathcal{F}.$$
\end{definition}

\begin{remark}\label{quasidom}
	We have the following property for all measures $P^{\theta}$, $\theta\in\Theta$, which can be seen as a weak form of absolute continuity with respect to $\P$: for each $A\in\mathcal{F}$ with $\P(A)=0$ there exists a subset $\Theta^A\subset \Theta$ with $\nu(\Theta^A)=1$ such that $P^{\theta}(A)=0$ for all $\theta\in\Theta^A$.
\end{remark}
Suppose we
are now given two initial believes on $(\Theta,\mathcal{D})$, i.e. two
different prior probability measures $\nu$, $\mu$. We introduce the respective mixture probability measures according to Definition
\ref{mixturemeas} for both the priors as follows
\begin{equation}\label{newmixtures}
	\PP_{\nu}(A):=\int_{\Theta}P^{\theta}(A)\nu(\mathrm{d}\theta), \qquad
	\PP_{\mu}(A):= \int_{\Theta}P^{\theta}(A)\mu(\mathrm{d}\theta),
\end{equation}
for all $A\in\mathcal{F}$.
\begin{lemma}\label{Lemma_changing_priors}
	Let $\mu$ and $\nu$ be two probability measures defined on the parameter space $(\Theta,\mathcal{D})$. Then,
	\begin{equation*}
		\mu \ll \nu \implies \PP_{\mu}\ll\PP_{\nu},
	\end{equation*}
	where $\PP_{\mu}$, $\PP_{\nu}$ are defined as in \eqref{newmixtures} and the last inequality holds for every $X\in\mathcal{X}$.
\end{lemma}
\begin{proof}
	If $\mu\ll\nu$, i.e. $\nu(D)=1$ implies $\mu(D)=1$
	for all $D\in\mathcal{D}$, then for any fixed $A\in\mathcal{F}$, if
	$\PP_{\nu}(A)=0$ there exists $\Theta^{A}\subset\Theta$ with
	$\nu(\Theta^{A})=1$ and $P^{\theta}(A)=0$ for all $\theta\in\Theta^{A}$. Hence
	\begin{equation*}
		\PP_{\mu}(A)=\int_{\Theta}P^{\theta}(A)\mu(\mathrm{d}\theta)=\int_{\Theta^{A}}P^{\theta}(A)\mu(\mathrm{d}\theta)=0,
	\end{equation*}
	which yields $\PP_{\mu}\ll\PP_{\nu}$.
\end{proof}
 
Clearly under the assumption $\mu\ll\nu$, we also have that $C^{\PP_{\nu}}\subseteq C^{\PP_{\mu}}$. Moreover if $\mu\sim\nu$ then by the previous result
$\PP_{\mu}\sim\PP_{\nu}$.
\begin{remark}
	Notice that by Lemma \ref{Lemma_changing_priors} and Remarks \ref{abs:fixset}, \ref{absconsP} we have that
	\begin{equation*}
		\rho^{\PP_{\mu}}(X)\le\rho^{\PP_{\nu}}(X), \  \quad\ \widehat{\rho}^{\PP_{\mu}}(X)\le\widehat{\rho}^{\PP_{\nu}}(X),
	\end{equation*}
	for all $X\in L^{\infty}(\PP_{\nu})$.
\end{remark}

\subsection{Classical representation for the mixture probability measure}

As in Section \ref{sec:classical} we discuss now the classical representation of robust risk measures for  $\cM^{\Theta}=\{P^{\theta} : \theta\in\Theta\}$, with respect to the mixture probability measure.

In order to prove similar results  in the case where $\mathcal{M}^{\Theta}=\{P^{\theta} : \theta\in\Theta\}$ is not supposed to be dominated a priori, we need some additional assumptions. Let the parameter set be a subset of a Polish space with $\mathcal{D}=\mathcal{B}(\Theta)$ the sigma-algebra of the Borel sets, and let $d(\cdot,\cdot)$ denote a distance on the Polish space. Our starting point builds on the (rather strong) assumption of continuity in total variation of the map $\theta\mapsto P^{\theta}$, which will  in turn imply that $\mathcal{M}^{\Theta}$ can be dominated, however not necessarily by the mixture probability measure. Indeed, the mixture probability measure will only serve as  dominating measure for a subset of  $\cM^{\Theta}$ induced from a subset of $\Theta $ that has $\nu$-measure $1$.
We then show by means of a specific example (see Example \ref{ex:coountercont}) that exactly this property  can also be obtained \emph{without}
the continuity assumption and \emph{without} the existence of a dominating measure for the whole set, implying that these are not
necessary conditions to obtain classical representations with respect to the mixture probability measure as of Theorem \ref{th:mainmixture}. 

Moreover, when specializing the setup to the case of law-invariance and considering the mixture probability measure $\mathbb{P}$ as the measure with respect to which  the penalty function  in \eqref{eq:alphacanonical1} is defined, we can prove a classical representation  with respect to $\mathbb{P}$ under mild conditions, see Corollary \ref{cor:final}.

Let us now start with the announced continuity assumption:
\begin{assumption}\label{asstvcon}
	We assume that the map $\theta\mapsto P^{\theta}$ is continuous from $\Theta\to\mathcal{M}_1$ with respect to the total variation norm, this means that for each $\ep>0$ there exists $\delta>0$ such that for $d(\theta,\theta')<\delta$ we have that
	$$
	\sup_{A\in\mathcal{F}}|P^{\theta}(A)-P^{\theta'}(A)|<\ep.
	$$
\end{assumption}

Recall Definition~\ref{mixturemeas} for the mixture probability measure: for each $A\in\mathcal{F}$
$$\P(A)=\int_{\Theta}P^{\theta}(A)\mathrm{d}\nu(\theta).$$

By Remark~\ref{quasidom} we know that for each $A\in\mathcal{F}$ with $\P(A)=0$ there exists $\Theta^A$ with $\nu(\Theta^A)=1$ such that $P^{\theta}(A)=0$ for all $\theta\in\Theta^A$. We would like to have the property that $P^{\theta}\ll \P$ for $\nu$-almost all $\theta$. The following example shows a phenomenon that we avoid by the previous continuity assumption.
\begin{example}
	Let $(\Omega,\mathcal{F})=([0,1], \mathcal{B}([0,1]))$.  Let $\Theta=[0,1]$ and define for each $\theta\in\Theta$ the measure $P^{\theta}=\delta_{\{\theta\}}$ (i.e. the Dirac measure in the point $\theta\in\Theta$). Define the set $c\mathcal{M}$ as the convex hull (with respect to even countable convex combinations) of these Dirac measures, i.e. $$c\mathcal{M}=\left\{\sum_{i=1}^{\infty}\alpha_i\delta_{\{\theta_i\}} : \text{  $0\leq\alpha_i\leq1$, \ with \ $\sum_{i=1}^{\infty}\alpha_i=1$\ and \ $\theta_i\in[0,1]$ for all $i\geq 1$}\right\}.$$ That means $c\cM=\overline{conv}(\cM^{\Theta})$. Let $\nu$ be the Lebesgue measure on $[0,1]$.  Then it holds that
	$$\P(A)=\int_{\theta\in[0,1]}P^{\theta}(A)\mathrm{d}\nu(\theta)=\int_{\theta\in[0,1]}\ind_{\{\theta\in A\}}\mathrm{d}\nu(\theta)=\nu(A).$$
	By Remark~\ref{quasidom} it holds that for each $A\in\mathcal{F}$ there exists $\Theta^A$ such that $\nu(\Theta^A)=1$ and $P^{\theta}(A)=0$ for all $\theta\in\Theta^A$. But each Dirac measure $P^{\theta}$ is singular to the Lebesgue measure and the same holds for each countable convex combination of Dirac measures. Obviously, $\theta\mapsto P^{\theta}=\delta_{\{\theta\}}$ is not continuous in $\theta$. Indeed, for each $\theta\in[0,1]$, we can find $\theta'\in[0,1]$ arbitrarily close to $\theta$ and $A\in\mathcal{F}$ such that $\theta\in A$, $\theta'\not\in A$ and hence $P^{\theta}(A)=1$ and $P^{\theta'}(A)=0$.
\end{example}

Note that in the previous example all measures $P^{\theta}$ are mutually singular. A similar phenomenon occurs for instance in the context of volatility uncertainty.

\begin{example}
Consider  stock price models with $\mathrm{d}S_t=S_t \theta \mathrm{d}W_t$, where $W$ is a standard Brownian motion and $\theta$ in $[0,1]$. Its laws denoted by $P^{\theta}$ are all mutually singular and thus Assumption \ref{asstvcon} cannot be satisfied.
\end{example}

Let us now show that there is a simple choice of a dominating measure for the set $\mathcal{M}^{\Theta}$.
\begin{lemma}\label{firstcountableprior}
Let $\Theta$ be a subset of a Polish space.
	Under Assumption~\ref{asstvcon} there exists a dominating measure $\mathbb{P}$ which is a countable convex combination of measures $P^{\theta}$, $\theta\in\Theta$.
\end{lemma}

\begin{proof}
By assumption there exists a countable dense subset of $\Theta$ which we denote by $\{\theta_k, k\in\mathbb{N}\}$.  Define $\P=\sum_{k=1}^{\infty}2^{-k}P^{\theta_k}$. This measure is in $c{\mathcal{M}^{\Theta}}$. We will show that this measure dominates $\mathcal{M}^{\Theta}$. Indeed, let $\P(A)=0$. Then $P^{\theta_k}(A)=0$ for all $k\in\mathbb{N}$. Fix an arbitrary $\theta\in\Theta$, then for each $\ep>0$ there exists $k\in\mathbb{N}$ such that $\theta_k$ is close enough to $\theta$ such that $|P^{\theta_k}(A)-P^{\theta}(A)|<\ep$. This follows by the density of the subset $\{\theta_k, k\in\mathbb{N}\}$ and by Assumption~\ref{asstvcon}. Hence, for each $\ep>0$ we get that $P^{\theta}(A)<\ep$ and so $P^{\theta}(A)=0$. Clearly $\P$ is even generalized equivalent with respect to $\mathcal{M}^{\Theta}$.

\end{proof}

The above lemma gives a certain prior $\nu$ on $\Theta$. Indeed it is the measure that gives the weight $2^{-k}$ to the parameter $\theta_k$. The measure $\P$ in Lemma \ref{firstcountableprior} above is the mixture probability measure for this prior.

Our goal is now to show a similar result for any mixture probability measure induced by  an arbitrary prior $\nu$.
Indeed, we can generalize the above result and show that under Assumption~\ref{asstvcon} the set  $\mathcal{M}^{\Theta}$ is $\nu$-a.s. dominated by the mixture probability measure.

Let us now state a relevant consequence of  Assumption~\ref{asstvcon}. Indeed  under this hypothesis, we can find set of measures that has $\nu$-measure $1$ and that can be  dominated by the mixture probability measure $\PP$.

\begin{lemma}\label{nu_a.s._dom}
	Let $\Theta$ be a subset of a Polish space and define $\Theta^0=\{\theta\in\Theta: P^{\theta}\not\ll\P\}$.
	Under Assumption~\ref{asstvcon} the subset $\Theta^0\subseteq\Theta$ is Borel-measurable and satisfies $\nu(\Theta^0)=0$.
\end{lemma}

\begin{proof}
	We assume without loss of generality $\Theta^{0}\neq\emptyset$. Indeed if $\Theta^{0}=\emptyset$, then it is trivially Borel-measurable and all $P^{\theta}\ll\PP$.
	\begin{enumerate}
		\item First we will show the measurability of $\Theta^{0}$.
		Define, for each $m\geq 1$,
		\begin{equation*}
			\Theta^m=\{\theta\in\Theta^0: \exists A\in\mathcal{F}\text{ with }\P(A)=0\text{ and }P^{\theta}(A)>2^{-m}\}.
		\end{equation*}
		Notice that since $\Theta^{0}\neq\emptyset$, then $\Theta^m\neq\emptyset$ for some $m\geq1$.
		It is clear that $\Theta^m\uparrow\Theta^0$. Moreover, each $\Theta^m$ is Borel-measurable. Indeed, fix $m\geq1$ and choose $\theta\in\Theta^m$. There exists $A$ with $\P(A)=0$ and $P^{\theta}(A)>2^{-m}$. Because of the strict inequality there exists $\ep>0$ such that $P^{\theta}(A)\geq2^{-m}+\ep$. Let us denote in the following with $B_{\delta}(\theta)$ the open ball with center $\theta\in\Theta$ and radius $\delta>0$. By Assumption~\ref{asstvcon} there exists $\delta>0$ small enough such that for all $\theta'\in B_{\delta}(\theta)$ we have that $|P^{\theta}(A)-P^{\theta'}(A)|<\frac{\ep}2$. Hence $P^{\theta'}(A)>2^{-m}+\ep-\frac{\ep}2>2^{-m}$ and each $P^{\theta'}\in\Theta^m$, i.e. $\Theta^m$ is open and therefore Borel-measurable. In particular we can conclude that $\Theta^{0}$ is open and thus Borel-measureable.
		
		\item We show now that $\nu(\Theta^{0})=0$.
		Suppose by contradiction that $\nu(\Theta^0)=\gamma>0$. By the above we have that $\nu(\Theta^0)=\lim_{m\to\infty}\nu(\Theta^m)$ and hence there exists $m_0\geq1$ such that $\nu(\Theta^{m_0})\geq\frac{\gamma}2$. Let $D$ be a dense countable subset of $\Theta$ and define $\tilde{\Theta}^{m_0}=\Theta^{m_0}\cap D$. Obviously this set is dense in $\Theta^{m_0}$ and countable. Let $(q_k)_{k\geq0}$ be any enumeration of this countable set, i.e.,
		$\tilde{\Theta}^{m_0}=\{q_k, k\geq1\}$. As $q_k\in\Theta^{m_0}$, $k\geq1$, we have that for each $k\geq1$ there exists $A_k\in\mathcal{F}$ such that $\P(A_k)=0$ and $P^{q_k}(A_k)>2^{-{m_0}}=:\epsilon$.
		Define
		$$A_0=\bigcup_{k=1}^\infty A_k\in\mathcal{F}.$$
		As $A_0$ is a countable union of $\P$-nullsets we still have that
		\begin{equation}\label{mixtnull}
			\P(A_0)=0.
		\end{equation}
		Choose an arbitrary $\theta\in\Theta^{m_0}$. By Assumption~\ref{asstvcon}, there exists $\delta>0$ such that for all $\theta'\in B_{\delta}(\theta)$ we have that
		\begin{equation}\label{uniformA}\sup_{A\in\mathcal{F}}|P^{\theta}(A)-P^{\theta'}(A)|<\frac{\ep}2.\end{equation}
		As  the sequence $(q_k)_{k\geq1}$ is dense there exists $k\ge1$ such that $q_{k}\in B_{\delta}(\theta)$. Since (\ref{uniformA}) holds for all $A\in\mathcal{F}$ we have that $P^{\theta}(A_k)\geq P^{q_k}(A_k)-\frac{\ep}2=\frac{\ep}2$.
		Hence $P^{\theta}(A_0)\geq P^{\theta}(A_k)\geq\frac{\epsilon}2$.
		As $\theta$ was arbitrary, we have that, for all $\theta\in\Theta^{m_0}$,
		$$P^{\theta}(A_0)\geq \frac{\epsilon}2.$$
		Let us now use the definition of the mixture probability measure to calculate $\P(A_0)$:

		$$
		\P(A_0) = \int_{\Theta}P^{\theta}(A_0)\mathrm{d}\nu(\theta)
		\geq\int_{\Theta^{m_0}}P^{\theta}(A_0)\mathrm{d}\nu(\theta)
		\geq \frac{\ep}2\nu(\Theta^{m_0})\geq\frac{\ep}2\cdot\frac{\gamma}2>0,
		$$
		which is a contradiction to (\ref{mixtnull}).
		
		Therefore we showed that $\nu(\Theta^0)=\nu(\{\theta\in\Theta: P^{\theta}\not\ll\P\})=0$.
	\end{enumerate}
\end{proof}

\begin{theorem}\label{th:mainmixture}
	Let $\Theta$ be a subset of a Polish space and let Assumption~\ref{asstvcon} hold.
Then there exists a Borel-measurable subset $\Theta^1\subset\Theta$ with $\nu(\Theta^1)=1$ such that, for all $\theta\in\Theta^1$, $P^{\theta}\ll\P$, where $\mathbb{P}$ denotes the mixture probability measure introduced in Definition \ref{mixturemeas}. Let $c\cM^1=\overline{\text{conv}}(\cM^{\Theta^1})$ where $\cM^{\Theta^1}=\{P^{\theta}: \theta\in\Theta^1\}$.
	Then, for all $X\in L^{\infty}(\P)$ it holds that
	\begin{equation*}
		\widehat{\rho}^{c\cM^1}(X):=\sup_{P \in c\cM^1} \widehat{\rho}^{P}(X)=\widehat{\rho}^{\P}(X).
	\end{equation*}
	and 
	\begin{equation*}
	\rho^{c\cM^1}(X):=\sup_{P \in c\cM^1} \rho^{P}(X)=\rho^{\P}(X).
\end{equation*}
\end{theorem}

\begin{proof}
This is a simple consequence of Lemma~\ref{nu_a.s._dom},  Theorem~\ref{th:nonhat} and Theorem~\ref{prop:criteria}, because  $\P$ is generalized equivalent to $c\cM^1$. Indeed, as we restrict ourselves to $\Theta^1$, we know that $P^{\theta}\ll\P$ for all $\theta\in\Theta^1$ (which satisfies $\nu(\Theta^1)=1$ by Lemma~\ref{nu_a.s._dom}). Clearly, if $A$ is such that $P^{\theta}(A)=0$ for all $\theta\in\Theta^1$ then $\P(A)=0$, implying that also the second property of generalized equivalence is satisfied. Theorem~\ref{prop:criteria} and Theorem~\ref{th:nonhat} now give the existence of a probability measure $P^0\in c\cM^1$ with $\P\ll P^0$. Hence the statements follow.
\end{proof}

\begin{remark} \label{rem:dommeasures}
\begin{enumerate}
    \item[(i)] 
Note that for the measure $P^0 \in c\cM^1 $  as introduced in the proof of Theorem~\ref{th:nonhat} and Theorem \ref{prop:criteria}, it holds that $P^0 \sim \P$ and thus by Remarks~\ref{abs:fixset} and ~\ref{absconsP}
$$\widehat{\rho}^{P^0}(X)=\widehat{\rho}^{\P}(X), \quad \text{and} \quad \rho^{P^0}(X)=\rho^{\P}(X).$$
In other words, under Assumption~\ref{asstvcon} there is a  measure in $c\cM^1 $ equivalent to $\P$ yielding the same risk. This measure is a countable convex combination of measures with parameters in $\Theta^1$ with $\nu(\Theta^1)=1$. Comparing this to the first countable convex combination of Lemma~\ref{firstcountableprior} the difference here is that $P^0$ can be interpreted as a countable new prior that fits to our prior $\nu$. So, if we think of a dynamical procedure, where we improve our knowledge using data this goes into the prior $\nu$. With the corresponding $P^0$ we then find a tailor-made countable prior. As the set $\Theta^1$ might be smaller than $\Theta$ the risk measure with respect to $P^0$ might therefore yield a smaller risk than the first countable convex combination of Lemma~\ref{firstcountableprior}.
\item[(ii)] Note that the mixture probability measure $\P$ does not necessarily lie in $c\cM^1 $, see e.g. Example \ref{P_not_cM} below. Hence the measure $P^0$ we find can be understood as a measure with respect to a countable prior that gives the same risk. By a countable prior $\tilde{\nu}$ for $P^0$ we mean a probability measure of the form
$\tilde{\nu}=\sum_{k=1}^{\infty}\gamma^k\delta_{\theta^k}$ such that $\gamma^k=\tilde{\nu}(\theta^k)$ and $P^0=\sum_{k=1}^{\infty}\gamma^kP^{\theta^k}$.
\end{enumerate}
\end{remark}

The following example shows that Assumption~\ref{asstvcon} is only sufficient to the get the  assertion of Theorem \ref{th:mainmixture}. Indeed, in this example we still get that the mixture probability measure is  a dominating measure for $\nu$-almost all $\theta$ even though  Assumption~\ref{asstvcon} is not satisfied.

\begin{example}\label{ex:coountercont}
The following example satisfies

\begin{enumerate}
\item There does not exist a dominating measure for $\{P^{\theta}: \theta\in\Theta\}$.
\item There exists $\widetilde{\Theta}\subseteq\Theta$ such that $\nu(\widetilde{\Theta})=1$ and such that there exists a dominating measure $\P$ for  $\mathcal{M}^{\widetilde{\Theta}}:=\{P^{\theta}: \theta\in\widetilde{\Theta}\}$. This dominating measure is also generalized equivalent with respect to $\mathcal{M}^{\widetilde{\Theta}}$. In particular, $\P$ can be chosen to be the mixture probability measure.
\item Assumption~\ref{asstvcon} does not hold on.
\item There are uncountably many $P^{\theta}$ that are mutually singular.
\end{enumerate}

Hence we see that we can have a $\nu$-a.s. dominating measure also in the case that Assumption~\ref{asstvcon} is not satisfied. Moreover a closer look at the example shows that we can have uncountably many singular measures as long as the corresponding parameters are in a $\nu$-nullset $\Theta\setminus\widetilde{\Theta}$. Moreover, the example shows that if there are uncountably many $P^{\theta}$ that are mutually singular we cannot find a dominating measure \emph{ for all} $P^\theta$, $\theta\in\Theta$ (only on a set of $\nu$-measure 1).
\newline\newline
The construction is as follows. Let $(\Theta,\mathcal{D},\nu)$ be $([0,1],\mathcal{B}([0,1]), \Lambda)$ where $\Lambda$ is again the Lebesgue-measure on $[0,1]$. $F\subset[0,1]$ be the Cantor set. It is well-known that $F$ is uncountable and $\Lambda(F)=0$. Define the measures $P^{\theta}$, $\theta\in[0,1]$ as follows:

$$P^{\theta}=\begin{cases} \delta_{\{\theta\}}&\text{for $\theta\in F$}\\
R^{\theta}&\text{for $\theta\in [0,1]\setminus F$,}
\end{cases}
$$
where $\frac{\mathrm{d}R^{\theta}}{\mathrm{d}\Lambda}=\frac1{\theta}\ind_{[0,\theta]}$. Note that $0\in F$, hence the $R^{\theta}$ are well-defined. We define $\widetilde{\Theta}=[0,1]\setminus F$.  Obviously $\nu(\widetilde{\Theta})=\Lambda(\widetilde{\Theta})=1$. Let us check the properties.

Concerning (ii) an obvious dominating measure is the Lebesgue measure $\Lambda$ as all $R^{\theta}\ll \Lambda$. Observe that the mixture probability measure $\P$ is equivalent to $\Lambda$. Indeed since all $R^{\theta}\ll \Lambda$ and $\Lambda(F)=0$, then clearly $\PP \ll \Lambda$. Notice that for every $A\in\mathcal{F}$,
\begin{align*}
	\PP(A)=\int_{\Theta}P^{\theta}(A)\nu(\mathrm{d}\theta)=\int_{0}^{1}R^{\theta}(A)\mathrm{d}\theta&=\int_{0}^{1}\int_{A}\frac{1}{\theta}\ind_{[0,\theta]}(s)\mathrm{d}s\mathrm{d}\theta\\
	&=\int_{A}\int_{s}^{1}\frac{1}{\theta}\mathrm{d}\theta\mathrm{d}s\\
	&=\int_{A}-\ln(s)\mathrm{d}s,
\end{align*}
therefore if $\PP(A)=0$, since the integrand is positive we obtain that $\Lambda(A)=0$. Hence we can also use $\P$ as dominating measure. It is also straightforward to see that $\mathcal{M}^{\widetilde{\Theta}}\ll \Lambda$. Indeed, if $P^{\theta}(A)=0$ for all $\theta\in\widetilde{\Theta}$, then we have that $\Lambda(A\cap[0,\theta])=0$ for all $\theta\in\widetilde{\Theta}$. Because $\nu(\widetilde{\Theta})=1$ we get that $\Lambda(A)=0$ as we can choose $\theta\in \widetilde{\Theta}$ arbitrarily close to 1.

For (iii) let $\theta'\in F$ and $\theta\in[0,1]\setminus F$ such that $|\theta-\theta'|<\delta$ with $\delta$ small. If $\theta<\theta'$ we can choose $A=[0,\theta]$, then $P^{\theta}(A)=R^{\theta}(A)=1$ and $P^{\theta'}(A)=\delta_{\{\theta'\}}(A)=0$. If $\theta'<\theta$ we can choose $A=\{\theta'\}$ to see that $P^{\theta}(A)=0$ whereas $P^{\theta'}(A)=1$. This works for arbitrarily small $\delta$. If we choose $\theta$ and $\theta'$ both in $F$ then the mutual singularity of the measures gives a similar result. Therefore the function $\theta\mapsto P^{\theta}$ cannot be continuous in total variation.

(iv) is obvious as $F$ is uncountable.

Let us now give a proof of (i). Suppose there would exists a dominating measure $R$ for all $P^{\theta}$, $\theta\in\Theta$. In particular we would have $P^{\theta}\ll R$ for all $\theta\in F$. Suppose $A$ is a nullset of $R$, i.e. $R(A)=0$. Suppose $\theta\in A\cap F$. Clearly, $P^{\theta}(A)=\delta_{\theta}(A)=1$, which is a contradiction to $P^{\theta}(A)=0$. It follows that if $R(A)=0$ then $A\cap F=\emptyset$. This implies that for all $x\in F$ we have that $R(\{x\})>0$. Such a probability measure cannot exist. Indeed, for all $x\in F$ there exists $\ep>0$ such that $R(\{x\})\geq \ep$. Fix $k\geq 1$ and define $F_k:=\{x\in F: R(\{x\})\geq 2^{-k}\}$. As $R$ is a probability measure there cannot be more than $2^k$ points in $F_k$, i.e. $|F_k|\leq 2^k$. Define $\widetilde{F}=\bigcup_{k\geq1}F_k$. This set is at most countable as it is a countable union of finite sets. Hence $F\setminus\widetilde{F}\not=\emptyset$. Now, take any $x\in F\setminus\widetilde{F}$. We have that $R(\{x\})<2^{-k}$ for all $k\geq 1$ and it follows that $R(\{x\})=0$, a contradiction.
This proof can be reproduced whenever we have uncountably many singular measures.

\end{example}

As already indicated in Remark \ref{rem:dommeasures}, the following example shows that the mixture probability measure does not necessarily have to lie in $c\cM$.

\begin{example}\label{P_not_cM}

We consider the previous example without the Cantor set, i.e.,
$(\Theta,\mathcal{D},\nu)=([0,1],\mathcal{B}([0,1]), \Lambda)$ where $\Lambda$ is again the Lebesgue-measure on $[0,1]$.  Define the measures $P^{\theta}$, $\theta\in[0,1]$ as the measures $R^{\theta}$ from before, i.e.,
$\frac{\mathrm{d}P^{\theta}}{\mathrm{d}\Lambda}=\frac1{\theta}\ind_{[0,\theta]}$, but now for all $\theta\in(0,1]$. Let $\mathcal{M}=\{P^{\theta}: \theta\in(0,1]\}$. As before the Lebesgue measure $\Lambda$ is a dominating measures for $\mathcal{M}$. The set $c\cM$ consists again of all countable convex combination of measures in $\mathcal{M}$. As, trivially, for $\theta=1$, $P^1=\Lambda\in\mathcal{M}$ and as it is equivalent to itself there exists a measure in $c\cM$ that is equivalent to $\Lambda$.
The mixture probability measure $\P$ satisfies, as before, that
  $$\PP(A)=\int_{A}-\ln(s)\mathrm{d}s,$$
and so it is equivalent to $\Lambda$. But it is not an element of $c\cM$ as we will show now.

Suppose we could find a countable convex combination $R=\sum_{n=1}^{\infty}\alpha_nP^{\theta^n}$ such that $R=\P$. Then we would have that $R$ and $\P$ agree on all sets $[0,c)$ with $c\in(0,1)$. Observe that
$$\P([0,c))=\int_0^c-\ln(s))\mathrm{d}s=c(1-\ln(c))=:f(c),$$
where $f(c)$ is a strictly concave increasing differentiable function on $(0,1)$. On the other hand $R((0,c])$ can be expressed as follows
 
 $$R([0,c))=\sum_{n=1}^{\infty}\alpha_nP^{\theta^n}([0,c))=\sum_{n=1}^{\infty}\alpha_ng^n(c)=:g(c),$$
 
 where $g^n(c)=\frac{c\wedge\theta^n}{\theta^n}$. In order to show that the two probability measures are equal we would have to show that $f(c)=g(c)$ for all $c\in(0,1)$.
\newline\newline
We will get a contradiction. Indeed, we will find $c_0$ with $0<c_0<1$ such that $g$ is not differentiable in $c_0$. Then $f$ and $g$ cannot agree for all $c\in(0,1)$ as $f$ is differentiable. 
\newline\newline
If $R$ should equal $\P$ then there has to exist $n_0$ in the indices of the convex combination such that $\alpha_{n_0}>0$ for a $0<\theta_{n_0}<1$. Indeed, otherwise $R=\Lambda$ and this is not equal to $\P$. Let $c_0=\theta^{n_0}$. Moreover, not all the weights for the indices $n$ with $\theta^n>c_0$ can be zero. Otherwise the support of $R$ would be $[0,c_0]$, and then $R$ could not be equal to $\P$ which has full support $[0,1]$.  Hence it has to hold that $\sum_{n\geq1, \theta^n>c_0}\alpha^n=\delta$, for some $\delta>0$.
\newline\newline
Clearly, by definition of $c_0$, the function $g^{n_0}$ is not differentiable in $c_0$. Note that the right hand derivative of $g^{n_0}$ in $c_0$ is equal to $0$, whereas the left hand derivative is equal to $\frac1{\theta^{n_0}}=\frac1{c_0}$.
All other functions $g^{\theta_n}$, $n\geq1$, $n\ne n_0$, are differentiable in $c_0$, with 

$${g^{n}}'(c_0)=\begin{cases}0&\text{for those $n$ with $\theta^n<c_0$},\\
\frac1{\theta^n}&\text{for those $n$ with  $\theta^n>c_0$}.\end{cases}$$
Recall that $g(c)=\sum_{n=1}^{\infty}\alpha_ng^{n}(c)$. Then, we get:

$$\lim_{h\downarrow 0}\frac{g(c_0+h)-g(c_0)}{h} =\sum_{n,\theta^n>c_0}\alpha_n\frac1{\theta_n}=:A,$$
where $A$ is a constant with $0<A<\frac{\delta}{c_0}$.

On the other hand
$$\lim_{h\uparrow 0}\frac{g(c_0+h)-g(c_0)}{h} =\frac{\alpha_{n_0}}{\theta^{n_0}}+\sum_{n,\theta^n>c_0}\alpha_n\frac1{\theta_n}=\frac{\alpha_{n_0}}{c_0}+A,$$
which is different from $A$ as $\alpha_{n_0}>0$. Hence $g$ is not differentiable in $c_0$.

Note that the interchange of limits in the above calculations is justified for $h\downarrow0$ as well as $h\uparrow0$.
To show this, introduce an artificial probability space $\tilde{\Omega}=\mathbb{N}$ with the probability measure that gives, to each $n\geq1$, the weight $\tilde{P}(n)=\alpha_n$. Define, for $h$, a random variable $X_h$ on $\tilde{\Omega}$ as $X_h(n)=\frac{g^n(c_0+h)-g^n(c_0)}{h}$. Hence, it holds, for all $n\not= n_0$ that $\lim_{h\to0}X_h(n)={g^n}'(c_0)$. For $n_0$ we have that $\lim_{h\downarrow0}X_h(n_0)=0$ and $\lim_{h\uparrow 0}X_h(n_0)=\frac1{c_0}$. Define $X(n)={g^n}'(c_0)$ for $n\not= n_0$ and $X(n_0)=0$. Then, pointwise, $X_h\to X$ for $h\downarrow0$ and $X_h\to X+\frac1{c_0}\ind_{n_0}$ for $h\uparrow0$. Moreover, it is easy to see that $|X_h|\leq\frac1{c_0}$ for all $h>0$ and  $|X_h|\leq\frac2{c_0}$, for all $h<0$ small enough such that $c_0+h>\frac{c_0}2$. Then, we can interpret the above interchange of limits as applications of the dominated convergence theorem for the measure $\tilde{P}$. Indeed,
$$\lim_{h\downarrow0}\frac{g(c_0+h)-g(c_0)}{h}=\lim_{h\downarrow0}\E_{\tilde{P}}[X_h]=\E_{\tilde{P}}[X]=A,$$
with $A$ as above and
$$\lim_{h\uparrow0}\frac{g(c_0+h)-g(c_0)}{h}=\lim_{h\uparrow0}\E_{\tilde{P}}[X_h]=\E_{\tilde{P}}[X+\frac1{c_0}\ind_{n_0}]=A+
\frac{\alpha_{n_0}}{c_0}.$$

\end{example}

Finally, we specialize our setup to the one of convex law-invariant risk measures as of Section \ref{sec:comparisionlawinv} and consider the mixture probability measure $\mathbb{P}$ as the measure with respect to which  the penalty function  in \eqref{eq:alphacanonical1} is defined.   Just by requiring that $\mathcal{M}^{\Theta}$ is closed under countable convex combinations and that the mass of sets of positive measure under $\mathbb{P}$ do not only stem from singular parts (see Condition \ref{eq:condA} below), we obtain the following result.

\begin{corollary} \label{cor:final}
Let $\cM^{\Theta} \subseteq \mathcal{M}_1$ be a set of probability measures induced by a parameter space $\Theta$ and closed under countable convex combinations. Let $\mathbb{P} $
 denote the mixture probability measure as of Definition \ref{mixturemeas} and consider law-invariance as in Section \ref{sec:comparisionlawinv} with respect to it, i.e. fix a penalty function $\alpha: \mathcal{M}_1 \to \mathbb{R} \cup \{+\infty\}$ as in \eqref{eq:alphacanonical1} with respect to $\mathbb{P}$. Suppose additionally that for all $A \in \mathcal{F}$ with $\mathbb{P}(A) >0$
 \begin{align}\label{eq:condA}
 \mathbb{P}(A)> \int_{\Theta} \widetilde{P}^{\theta}_2 (A) \nu(\mathrm{d}\theta),
 \end{align}
 where $\widetilde{P}^{\theta}_2$ denotes  the (unnormalized) singular part of the Lebesgue decomposition of $P^{\theta}$ with respect to $\mathbb{P}$. 
 Then for all $X \in \mathcal{X}^{\mathcal{M}^{\Theta}} \cap L^{\infty}(\PP)$ we have
	$$
		\widehat{\rho}^\mathbb{\mathcal{M}}(X)=\widehat{\rho}^\mathbb{P}(X)= \rho^\mathbb{P}(X)=\rho^\mathbb{\mathcal{M}}(X).
		$$
\end{corollary}

\begin{proof}
We start by proving $\rho^\mathbb{P}=\rho^\mathbb{\mathcal{M}}$.
Note that the only difference with respect to Corollary \ref{cor:worstcaseprob1} is that $\mathbb{P}$ does not need to lie in $\mathcal{M}^{\Theta}$. Instead of this, we  assumed that $\mathbb{P}$ is the mixture probability measure and that $\cM^{\Theta}$ is closed under countable convex combinations, the latter with the goal to apply Theorem \ref{hs}. To this end we need additionally a family of probability measures which is absolutely continuous with respect to $\mathbb{P}$. Consider therefore instead of $P^{\theta}$ always $P^{\theta}_1$ denoting  the normalized absolutely continuous part of the Lebesgue decomposition with respect to $\mathbb{P}$. Then, as visible from the proof of Proposition \ref{prop:PvsR} (see also Remark \ref{remark:orderetal} (iii))
\[
\rho^{P^{\theta}_1}=\rho^{P^{\theta}}
\]
and thus $\rho^{\mathcal{M}^{\Theta}}= \sup_{P^{\theta}_1} \rho^{P^{\theta}_1}$ on $\mathcal{X}^{\mathcal{M}^{\Theta}} $.
Hence is suffices to consider the family $ \mathcal{M}^{\Theta}_1=\{P^{\theta}_1  : \theta \in \Theta\}$, which inherits the property of being closed under countable convex combinations.
In order to apply Theorem \ref{hs} we need additionally that the condition $P^{\theta}_1 (A)=0$ for all  $P^{\theta}_1 \in \mathcal{M}^{\Theta}_1$ implies  $\mathbb{P}(A)=0$.
Denoting by $\widetilde{P}^{\theta}_i$ the unnormalized parts of the Lebesgue decomposition of $P^{\theta}$ with respect to $\mathbb{P}$, we have
\[
\mathbb{P}(A)= \int_{\Theta} ( \widetilde{P}^{\theta}_1(A)+ \widetilde{P}^{\theta}_2(A)) \nu(\mathrm{d} \theta)=\int_{\Theta}   \widetilde{P}^{\theta}_2(A) \nu(\mathrm{d} \theta),
\]
where the last equality holds 
if
$P^{\theta}_1 (A)=0$ for all  $P^{\theta}_1 \in \mathcal{M}^{\Theta}_1$. 
Therefore, Condition \ref{eq:condA} yields $\mathbb{P}(A)=0$.
Theorem \ref{hs} thus implies that there exists some $P_0 \in \mathcal{M}^{\Theta}_1$ which is equivalent  to $\mathbb{P}$ and hence $\rho^{\mathbb{P}}= \rho^{P_0}$.
Moreover, using $P_0 \in \mathcal{M}^{\Theta}_1$ and Proposition  \ref{prop:PvsR}, we get
\[
 \rho^{P_0} \leq \rho^{\mathcal{M}^{\Theta}_1}=\rho^{\mathcal{M}^{\Theta}} \leq \rho^{\mathbb{P}}.
\]
As $\rho^{\mathbb{P}}= \rho^{P_0}$, all inequalities reduce to equalities. The same reasoning holds for $\widehat{\rho}$ and we can thus conclude since $\rho^{\mathbb{P}}= \widehat{\rho}^{\mathbb{P}} $ again by Proposition \ref{prop:PvsR}.
\end{proof}


\end{document}